\documentclass[11pt]{article}
\usepackage{times}
\usepackage[margin=1in]{geometry}
\usepackage{graphicx}
\usepackage{enumerate}
\usepackage{rotating}
\usepackage{verbatim}
\usepackage{amsmath}
\usepackage{picins,boxedminipage}
\usepackage{amsfonts}
\usepackage{setspace}
\usepackage{pstricks}
\usepackage{pst-node}
\usepackage{pst-plot}
\usepackage{pst-tree}
\usepackage{floatflt,url,hyperref}
\usepackage{txfonts}  % Set math in a Times font, since the text is in Times
\newtheorem{theorem}{Theorem}[section]
\newtheorem{definition}[theorem]{Definition}
\newtheorem{lemma}[theorem]{Lemma}
\newtheorem{conjecture}[theorem]{Conjecture}

\newtheorem{observation}[theorem]{Observation}
\newtheorem{claim}[theorem]{Claim}

\newcommand{\testbf}[1]{\mathbf{#1}}
\newcommand{\qed}{\mbox{}\hspace*{\fill}\nolinebreak\mbox{\rule{6pt}{6pt}}}
\newenvironment{proof}{\vspace{-2mm}\noindent {\bf Proof:}}{\qed\par\medskip}

\begin{document}
\title{When the Cut Condition is Enough: A Complete Characterization
  for Multiflow Problems in Series-Parallel Networks\footnote{This
    work is supported in part by NSF grants CCF-0728869,
    CCF-1016778, and IIS-0916565.}}

\author{
Amit Chakrabarti
\and Lisa Fleischer
\and Christophe Weibel
}

%\thanks{ Research supported in part by: NSERC}
\maketitle
\def\thepage {} % Kill pagenumbering
\thispagestyle{empty}

% ---> Now read in the various sections <---

% ---> The article's abstract <---

\begin{abstract}
  Let $G=(V,E)$ be a supply graph and $H=(V,F)$ a demand graph defined
  on the same set of vertices. An assignment of capacities to the
  edges of $G$ and demands to the edges of $H$ is said to satisfy the
  \emph{cut condition} if for any cut in the graph, the total demand
  crossing the cut is no more than the total capacity crossing it.
  The pair $(G,H)$ is called \emph{cut-sufficient} if for any
  assignment of capacities and demands that satisfy the cut condition,
  there is a multiflow routing the demands defined on $H$ within the
  network with capacities defined on $G$.

  We prove a previous conjecture, which states that when the supply
  graph $G$ is series-parallel, the pair $(G,H)$ is cut-sufficient if
  and only if $(G,H)$ does not contain an \emph{odd spindle} as a minor;
  that is, if it is impossible to contract edges of $G$ and delete edges
  of $G$ and $H$ so that $G$ becomes the complete bipartite graph
  $K_{2,p}$, with $p\geq 3$ odd, and $H$ is composed of a cycle
  connecting the $p$ vertices of degree $2$, and an edge connecting the
  two vertices of degree $p$. We further prove that if the instance is
  \emph{Eulerian} --- that is, the demands and capacities are integers
  and the total of demands and capacities incident to each vertex is
  even --- then the multiflow problem has an integral solution. We
  provide a polynomial-time algorithm to find an integral solution in
  this case.

  In order to prove these results, we formulate properties of tight
  cuts (cuts for which the cut condition inequality is tight) in
  cut-sufficient pairs. We believe these properties might be useful in
  extending our results to planar graphs.
\end{abstract}
\begin{figure}[h]
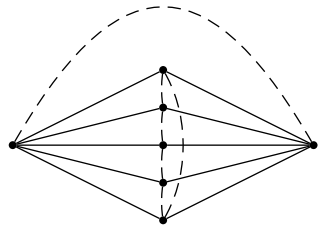

\begin{center}
\psset{unit=0.5cm,xunit=2.0cm,arrows=-,shortput=nab,linewidth=0.5pt,arrowsize=2pt 5,labelsep=1.5pt}
\pspicture(0,0)(2,5)
\dotnode(0,2){u}
\dotnode(1,0){a}
\dotnode(1,1){b}
\dotnode(1,2){c}
\dotnode(1,3){d}
\dotnode(1,4){e}
\dotnode(2,2){v}
\ncline{u}{a}
\ncline{u}{b}
\ncline{u}{c}
\ncline{u}{d}
\ncline{u}{e}
\ncline{v}{a}
\ncline{v}{b}
\ncline{v}{c}
\ncline{v}{d}
\ncline{v}{e}
\psset{linestyle=dashed}
\ncarc{a}{b}
\ncarc{b}{c}
\ncarc{c}{d}
\ncarc{d}{e}
\ncarc[arcangle=30]{e}{a}
\ncarc[arcangle=60,ncurv=1.4]{u}{v}
\endpspicture
\caption{A $5$-spindle. Supply edges are solid and demand edges are dashed.}
\label{fig:abstract}
\end{center}
\end{figure}
\vspace{2cm}
\clearpage
\pagenumbering{arabic}

% ---> Overview of the paper, relation to previous work <---

\section{Introduction} \label{sec:intro}

When does a network admit a flow that satisfies a given collection of
point-to-point demands? This broad question has led to a number of
important results over the last several decades. The most fundamental of
these considers the case of a {\em single} demand, from a source vertex
to a sink vertex. In this case, the network is able to satisfy the
demand if and only if for every cut separating the source from the sink,
the total capacity of network edges crossing the cut is no less than the
demand: this holds regardless of the topology of the network. This is
the famous max-flow min-cut theorem, celebrated both for its elegance
and its very wide applicability across computer science, graph theory,
and operations research.

Things get much more interesting, and intricate, when we generalize to
the multicommodity case. It is easy to see that in order to have a flow
satisfying all demands, it is necessary that for all cuts the total
capacity crossing the cut is no less than the total demand crossing it.
This is called the {\em cut condition}.  Unlike in the single-commodity
case, this is no longer a sufficient condition in
general~\cite{Okamura81}. This has led to two kinds of generalizations:
(1) finding conditions on the topology of the network and/or the
structure of the demands that make the cut condition sufficient, and (2)
understanding how ``far'' from sufficient the cut condition can be. We
shall discuss both categories of results below, after the necessary
basic definitions. The work presented in this paper falls into the first
category.

The simplest example demonstrating that the cut condition does not
suffice is the network $K_{2,3}$, with unit capacities and unit demand
between each pair of non-adjacent vertices. This example has a natural
generalization to the network $K_{2,p}$ for odd $p \ge 3$, as
suggested by Figure~\ref{fig:abstract}. Our main theorem says that for
the important class of series-parallel networks, these examples ---
which we call {\em odd spindles} --- are (in a sense) the {\em only}
ones where the cut condition does not suffice.

The single-commodity flow problem has another nice property that does
not extend to the multicommodity case: if the demand and the capacities
are integers and the network can satisfy the demand, then it can do so
with an {\em integral} flow. In our work, we show that integral
multicommodity flow instances on series-parallel networks that satisfy
the cut condition and avoid the above odd spindles admit {\em
half-integral} flows satisfying the demands (in fact, we show a stronger
result which implies this; see below). Moreover, for such instances, we
give a polynomial time algorithm to compute such a flow.

\subsection{Basic Definitions and Background}

Given an undirected graph $G=(V,E)$, with capacities $c_e$ on the edges
$e\in E$, let $\mathcal{P}$ be the set of simple paths in $G$.  A {\em
multiflow} is an assignment $f:\mathcal{P}\rightarrow\mathbb{R_+}$. It
is said to be {\em feasible} if, for each $e\in E$, we have $\sum_{P\in
\mathcal{P}(e)}f_P \le c_e$, where $\mathcal{P}(e)$ is the set of paths
in $G$ that contain the edge $e$.  Let $H=(V,F)$ be another graph on the
same set of vertices, with demands $D_i$ on the edges $i\in F$.
The multiflow $f$ is said to {\em satisfy} $H$ if for each edge $i\in
F$, we have $\sum_{P\in \mathcal{P}[i]}f_P\geq D_i$, where
$\mathcal{P}[i]$ is the set of paths in $G$ that connect the endpoints
of $i$.  The tuple $(G,H,c,D)$ forms an instance of the multiflow (or
multicommodity flow) problem, which consists of finding whether there
exists a feasible multiflow in $G$ satisfying $H$; if so, the instance
is said to be \emph{routable}. We call $G$ and $H$ the {\em supply
graph} and {\em demand graph} (respectively) of the instance.

For each set $C\subseteq V$, the {\em cut} $\delta_G(C)$ generated by
$C$ in $G$ is defined to be the set of edges in $G$ with exactly one
endpoint in $C$.  We define $\delta_H(C)$ similarly.  The {\em surplus}
$\sigma(C)$ of $C$ is the total capacity of the edges in $\delta_G(C)$
minus the total demand of the edges in $\delta_H(C)$:
$\sigma(C)=\sum_{e\in\delta_G(C)}c_e - \sum_{i\in\delta_H(C)}D_i$.  The
cut condition is then the statement that every cut has nonnegative
surplus: $\sigma(C)\geq 0$ for all $C\subseteq V$. As noted above, an
instance $(G,H,c,D)$ must satisfy the cut condition in order to be
routable. Our goal is to understand when this condition is sufficient.

The graph pair $(G,H)$ is {\em cut-sufficient} if for all assignments of
capacities $c$ and demands $D$ that satisfy the cut condition, the
instance $(G,H,c,D)$ is routable. One of the earliest cut-sufficiency
theorems is due to Hu~\cite{Hu63} and states that $(G,H)$ is
cut-sufficient if $H$ is the union of two stars, i.e., if all of its
edges can be covered by two vertices. Notice that this theorem applies
to a general $G$, but it greatly restricts $H$. Network flow literature
abounds with other cut-sufficiency
theorems~\cite{Hu63,Lomonosov85,Okamura83,Okamura81,Schrijver89,Seymour81a,Seymour81}.
Many of these impose conditions on both $G$ and $H$; a well-known
example is the Okamura-Seymour Theorem, which states that a pair $(G,H)$
is cut-sufficient if $G$ is planar and all edges of $H$ have their
endpoints on a single face of $G$~\cite{Okamura81}.
Schrijver~\cite[Chapter~70]{SchrijverBook} surveys several
cut-sufficiency theorems and many related concepts and topics.

\subsection{Our Contributions}

We give a sharp characterization of cut-sufficient graph pairs where the
supply graph is {\em series-parallel}. Further, for integral multiflow
instances on such cut-sufficient pairs, we show that the cut condition
together with a natural ``Eulerian'' condition imply that a feasible
integral solution exists; we also give a polynomial time algorithm to
find an integral solution. Finally, our work here suggests to us a
conjecture that would characterize cut-sufficiency in {\em planar}
graphs. The details follow.

We define a \emph{$p$-spindle} to be a pair of graphs $(G,H)$ such that
the supply graph $G$ is $K_{2,p}$, with $p \ge 3$, and the demand graph
$H$ consists of a cycle connecting the $p$ vertices of degree $2$ in
$G$, and an additional demand edge between the two remaining vertices.
An {\em odd spindle} is a $p$-spindle with $p$ odd.

\begin{theorem}[Fractional Routing Theorem; characterization of cut-sufficiency]\label{thm:fraction}
  If the supply graph $G$ is series-parallel, then the pair $(G,H)$ is
  cut-sufficient if and only if the pair $(G,H)$ cannot be reduced to an
  odd spindle by contraction of edges of $G$ and deletion of edges of
  $G$ and $H$.
\end{theorem}

Schrijver~\cite[Section~70.11]{SchrijverBook} gives a number of
sufficient conditions for cut-sufficiency; our characterization above is
sharper than all of these when $G$ is series-parallel. The above result
was conjectured in Chekuri et al.~\cite[Conjecture 3.5]{Chekuri10}. To
prove it, we first revisit the connection between multiflow problems and
metric embeddings via linear programming duality. Unlike previous works
that used this approach, we exploit {\em complementary slackness} to
derive some LP-based conditions for cut-sufficiency, in
Section~\ref{sec:new}.  These conditions do not refer to the structure
of $G$, and so could be useful in extending our results from
series-parallel graphs to more general classes.  The proof of
Theorem~\ref{thm:fraction} itself appears in Section~\ref{sec:fraction}. 

We say that an instance $(G,H,c,D)$ is \emph{Eulerian} if all capacities
$c_e$ and demands $D_i$ are integers and $\sigma(C)$ is even for all
$C\subseteq V$; recall that $\sigma$ depends on $c$ and $D$.
\begin{theorem}[Integral Routing Theorem]\label{thm:integral}
  If $G$ is series-parallel, $(G,H)$ is cut-sufficient, and the
  multiflow instance $(G,H,c,D)$ satisfies the cut condition and is
  Eulerian, then the problem has an integral solution.  Moreover, 
  under these conditions, an integral solution can be computed in 
  polynomial time.
\end{theorem}
This implies that under the same assumptions except for the Eulerian
condition, the multiflow problem has a half-integral solution. Similar
uses of the Eulerian condition are ubiquitous in the
literature~\cite{Okamura83,Okamura81,Schrijver89}. We prove the above
result in Section~\ref{sec:integral}.  The algorithm is described
in Section~\ref{sec:alg}.
% In case there is no integral solution, we also provide an algorithm
% finding the odd-$K_{2,p}$-pair. NO WE DON'T

Planar supply graphs allow one other obstruction to cut-sufficiency,
apart from the odd spindles. We conjecture, in
Section~\ref{sec:conclusion}, that there are no further examples: this
would extend our results to instances where $G$ is planar.

\subsection{Other Related Work}

A different approach to the relation between multiflows and cuts was
pioneered by Leighton and Rao~\cite{Leighton99}, who sought to
understand how ``far'' from sufficient the cut condition could be. To be
precise, let us define the {\em maximum concurrent flow} for a multiflow
instance $(G,H,c,D)$ to be the largest fraction $\phi$ such that
$(G,H,c,\phi D)$ is routable.  In this paper, we adopt the equivalent
approach of studying the \emph{minimum congestion} $\alpha \ge 1$ such
that $(G,H,\alpha c,D)$ is routable: it is easy to see that
$\phi=1/\alpha$ for any instance.  For a pair of graphs $(G,H)$, the
\emph{flow-cut gap} is defined as the maximum, over all choices of
demands and capacities that satisfy the cut condition, of the minimum
congestion. The larger this gap the further the pair $(G,H)$ is from
cut-sufficiency. Clearly, a pair is cut-sufficient if and only if its
flow-cut gap is $1$.

There has been intense research on finding the flow-cut gaps for various
classes of graphs, a line of work originally motivated by the problem of
approximating sparsest
cuts~\cite{Aumann98,Chakrabarti08,Gunluk07,Gupta99,Linial94}.  The class
of series-parallel instances is notable, as it is one of the very few
classes for which there are precise bounds on the flow-cut gap:
Chakrabarti et al.~\cite{Chakrabarti08} show that the gap cannot be more
than $2$, whereas Lee and Raghavendra~\cite{Lee07} show that it can be
as close to $2$ as desired. Chekuri et al.~\cite{Chekuri09} show that
series-parallel instances have \emph{integral multiflows} that do not
use more than $5$ times the capacity of the supply graph.
A special
case of the integer multiflow problem is the \emph{disjoint paths problem},
where $D_i = 1$ for all $i$ and $c_e = 1$ for all $e$.  
In general, the disjoint paths problem is NP-complete even when
restricted to series-parallel graphs~\cite{Vygen01}.

The seminal work of Linial et al.~\cite{Linial94} connected flow-cut
gaps to metric embeddings via LP duality: we now briefly explain this
connection, which we also use in our work. Every positive length
function $l$ on the edges of a graph determines a \emph{shortest-path
metric}, which is a distance function $d$ on the vertices of the graph,
such that $d(u,v)=\min_{P\in \mathcal{P}[u,v]}\sum_{e\in P}l_e$; here
$\mathcal{P}[u,v]$ denotes the set of paths between the vertices $u$ and
$v$. Every assignment of non-negative real values $x_C$ to subsets $C$ of
vertices of a graph determines a \emph{cut-cone metric}, which is a
distance function $d$ on the vertices of the graph defined by
$d(u,v)=\sum_{C:|\{u,v\}\cap C|=1}x_C$.  For any two distance functions
$d$ and $d'$ defined on the vertices of a graph such that $d\geq d'$,
the \emph{distortion} from $d$ to $d'$ is defined to be $\max_{u\ne v}
d(u,v)/d'(u,v)$.  For a distance function $d$ and a family of metrics
$\mathcal{M}$, the minimum distortion embedding of $d$ into
$\mathcal{M}$ is a distance function $d'$ in $\mathcal{M}$ that
minimizes the distortion from $d$ to $d'$.  Linial et
al.~\cite{Linial94} show that the maximum congestion required for a
particular supply graph $G$ equals the maximum distortion required to
embed any possible shortest-path metric on $G$ into the family of
cut-cone metrics.\footnote{It is a simple exercise to show that the
family of cut-cone metrics coincides with that of $\ell_1$-embeddable
metrics.} We shall call this the \emph{congestion-distortion
equivalence} theorem.

% In Section~\ref{sec:notation}, we define some important concepts and
% notations. In Section~\ref{sec:new}, we present a new proof of the
% connection between the flow-cut gap and minimum distortion embeddings,
% based on linear programming. In Section~\ref{sec:fraction}, we
% characterize cut-sufficient pairs $(G,H)$ for $G$ series-parallel. In
% Section~\ref{sec:integral}, we prove the existence of integral
% solutions when $G+H$ is Eulerian and provide an algorithm.  In
% Section~\ref{sec:conclusion}, we present a conjecture extending our
% results to instances where $G$ is planar.

% ---> The notation and definitions used in the article <---

\section{Definitions and Preliminaries}\label{sec:notation}

% Let $G=(V,E)$ be a supply graph, and $c_e$ denote a nonnegative
% capacity for any edge $e$ of $G$. Let $H=(V,F)$, and $D_i$ be a
% nonnegative demand for any edge of $H$. For any subset of vertices
% $C\subseteq V$, we denote as $\delta_G(C)$ the set of edges of $G$
% with exactly one endpoint in $C$, and as $\delta_H(C)$ the set of
% edges of $H$ with exactly one endpoint in $C$. Thus, each subset of
% vertices $C$ generates a cut in graphs $G$ and $H$. We call
% \emph{surplus} of $C$ the difference between the total capacity of
% edges in $\delta_G(C)$ and the total demand of edges in $\delta_H(C)$,
% and denote as
% $\sigma(C)=\sum_{e\in\delta_G(C)}c_e-\sum_{f\in\delta_H(C)}D_i$. In
% particular, the \emph{cut condition} can be written $\sigma(C)\geq 0$,
% for all $C\subseteq V$.

A subset of vertices $C\subseteq V$ and the corresponding cut
$\delta_G(C)$ are called \emph{central} if both $C$ and
$V\setminus C$ are connected in $G$. It is well-known
and easy to prove that if the surplus $\sigma$ is nonnegative for all
central cuts, then the cut condition is satisfied~\cite[Theorem
70.4]{SchrijverBook}. A subset $C$ and the cut $\delta_G(C)$ are
\emph{tight} if $\sigma(C)=0$.

We assume in this article that the supply graph $G$ is biconnected. It
is not hard to show that if $G$ is not biconnected, the multiflow
problem can be solved separately on its biconnected components. A
biconnected graph is \emph{series-parallel} if and only if it does not
contain $K_4$ as a minor. A pair of graphs $(G,H)$ is series-parallel if
the supply graph $G$ is series-parallel. 

We use an extension of graph minors to pairs $(G,H)$ of supply and
demand graph, as proposed in~\cite{Chekuri09}.

\begin{definition}Let $(G,H)$ and $(G',H')$ be two pairs of graphs.
  Then $(G',H')$ is a \emph{minor} of $(G,H)$ if we can obtain
  $(G',H')$ from $(G,H)$ by contracting and deleting edges of $G$, and
  deleting edges of $H$.
\end{definition}
Here, \emph{deleting} an edge means removing it from the graph, and
\emph{contracting} an edge means removing it and merging its
endpoints.

\subsection{Surplus Identities}

Recall that the \emph{surplus} $\sigma(X)$ of $X\subseteq V$ is the
total capacity minus the total demand crossing the cut $\delta_G(X)$.
Additionally, for $X$ and $Y$ disjoint, let $\delta_G(X,Y)$ and
$\delta_H(X,Y)$ be the set of edges in $G$, respectively $H$, with one
endpoint in $X$ and one in $Y$, and let $\sigma(X,Y):=\sum_{e\in
\delta_G(X,Y)}c_e-\sum_{e\in \delta_H(X,Y)}D_i$. In particular,
$\sigma(X,V\setminus X) = \sigma(X)$. The surplus function $\sigma$
satisfies the following useful identities.

\begin{lemma}\label{lem:easy1}
Let $A$ and $B$ be two subsets of $V$. Then \\
  $\langle a \rangle$ If $B_1,\ldots,B_k$ is a partition of $B$, then
  $\sigma(A,B)=\sigma(A,B_1)+\cdots+\sigma(A,B_k)$.\\
  $\langle b \rangle$ In particular, if $B=V\setminus A$, then
  $\sigma(A)=\sigma(A,B_1)+\cdots+\sigma(A,B_k).$\\
  $\langle c \rangle$  $   \sigma(A\cup B)+\sigma(A\cap B) 
    = \sigma(A)+\sigma(B)-2\sigma(A\setminus B,B\setminus A),$\\
    $\langle d \rangle$  $ \sigma(A\setminus B)+\sigma(B\setminus A) 
    = \sigma(A)+\sigma(B)-2\sigma(A\cap B,V\setminus (A\cup B)).$\\
  $\langle e \rangle$ In particular, if $A$ and $B$ are disjoint, then
  $\sigma(A\cup B) = \sigma(A)+\sigma(B)-2\sigma(A,B).$
\end{lemma}
\begin{proof}
  $\langle a \rangle$ and $\langle b \rangle$ are easy to prove, and
  are left as an exercise.  Let us use $\overline{X}$ to denote
  $V\setminus X$, for each subset $X\subseteq V$.  By
  $\langle a \rangle$ and  $\langle b \rangle$, we have
  \begin{align*}
  \sigma(A\cup B) &= \sigma(A\cup B,\overline{A\cup B})
    =\sigma(A\setminus B,\overline{A\cup B})
    +\sigma(A\cap B,\overline{A\cup B})+\sigma(B\setminus A,\overline{A\cup B}), \\
  \sigma(A\cap B) &= \sigma(A\cap B,\overline{A\cap B})
    =\sigma(A\cap B,\overline{A\cup B})+\sigma(A\cap B,A\setminus B)
    +\sigma(A\cap B,B\setminus A), \\
  \sigma(A) &= \sigma(A,\overline A)
    =\sigma(A\cap B,\overline{A\cup B})+\sigma(A\cap B,B\setminus A)
    +\sigma(A\setminus B,\overline{A\cup B})+\sigma(A\setminus B,B\setminus A), \\
  \sigma(B) &= \sigma(B,\overline B)
    =\sigma(A\cap B,\overline{A\cup B})+\sigma(A\cap B,A\setminus B)
    +\sigma(B\setminus A,\overline{A\cup B})+\sigma(B\setminus A,A\setminus B).
  \end{align*}
  Simplifying, we get $\langle c \rangle$. Additionally, we have
  \begin{align*}
  \sigma(A\setminus B) &= \sigma(A\setminus B,\overline{A\setminus B})
    =\sigma(A\setminus B,\overline{A\cup B})+\sigma(A\setminus B,B\setminus A)
    +\sigma(A\setminus B,A\cap B), \\
  \sigma(B\setminus A) &= \sigma(B\setminus A,\overline{B\setminus A})
    =\sigma(B\setminus A,\overline{A\cup B})+\sigma(B\setminus A,A\setminus B)
    +\sigma(B\setminus A,A\cap B).
  \end{align*}
  By comparing to the equations for $\sigma(A)$ and $\sigma(B)$, we
  get $\langle d \rangle$. Finally, $\langle e \rangle$ is just a
  restatement of $\langle c \rangle$ for $A$ and $B$ disjoint.
\end{proof}

\subsection{Properties of Biconnected and Series-Parallel Graphs}
\label{sec:graph-props}

We now establish a number of simple but useful properties of biconnected
and series-parallel graphs that arise at various points in the proofs
of our main theorems. The reader who wishes to focus on the main theorems
may safely skip to Section~\ref{sec:new}.

\begin{lemma}\label{lem:cycle}
  In a series-parallel graph, a simple cycle does not intersect any
  central cut more than twice.
\end{lemma}
\begin{proof}
  A cycle intersects any cut an even number of times. Suppose a cycle
  $Q$ intersects a cut $\delta_G(C)$ four times or more. Then we can
  choose four vertices $u_1,u_2,u_3,u_4$ in order on $Q$ such that
  $u_1$ and $u_3$ are on one side of $\delta_G(C)$ and $u_2$ and $u_4$
  on the other. Then there is a path connecting $u_1$ to $u_3$ on one
  side of $\delta_G(C)$, and a path connecting $u_2$ to $u_4$ on the
  other. This creates a $K_4$ minor, which cannot exist in a
  series-parallel graph.
\end{proof}

\begin{lemma}\label{lem:sut}
  In a biconnected graph $G$, for any three distinct vertices $s$, $u$ 
  and $t$, there is a simple path from $s$ to $t$ containing $u$.
\end{lemma}
\begin{proof}
  Since $G$ is biconnected, there are two vertex-disjoint paths $P_1$
  and $P_2$ from $s$ to $u$. Biconnectivity also implies there is a path
  $P$ from $u$ to $t$ disjoint from $s$. If $P$ does not intersect $P_1$
  (or $P_2$), then $P_1$ (or $P_2$) followed by $P$ creates a simple
  path connecting $s$--$u$--$t$ in that order. Otherwise, let $w$ be the last
  vertex of $P$, from $u$ to $t$, that is in $P_1$ or $P_2$. Since $P_1$
  and $P_2$ are disjoint, $w$ is in only one of them, say $P_2$. Then
  $P_1$, the part of $P_2$ from $u$ to $w$, and the part of $P$ from $w$
  to $t$ is simple, and connects $s$--$u$--$t$ in that order.
\end{proof}

\begin{lemma}\label{lem:suvt}
  In a biconnected graph $G$, for any two vertices $s$ and $t$ and any
  edge $(u,v)$, there is a simple path from $s$ to $t$ containing the
  edge $(u,v)$.
\end{lemma}
\begin{proof}
  By Lemma~\ref{lem:sut}, there are simple paths $P_{sut}$ and
  $P_{svt}$ from $s$ to $t$ containing $u$ and $v$ respectively. If
  $P_{sut}$ contains $v$, or $P_{svt}$ contains $u$, then using the
  edge $(u,v)$ to shortcut the path, we get a path from $s$ to $t$
  containing $(u,v)$. Let $P_{su}$, $P_{sv}$, $P_{ut}$ and $P_{vt}$ be
  the subpaths of $P_{sut}$ and $P_{svt}$ between corresponding
  vertices.
  % If $P_{su}$ and $P_{vt}$ do not intersect, or $P_{sv}$ and
  % $P_{ut}$ do not intersect, then they form a path from $s$ to $t$
  % containing $(u,v)$.
  In the set of vertices in $P_{sut}\cap P_{svt}$, let $w$ be a vertex
  closest to $u$ on $P_{sut}$ (in $P_{su}$ or $P_{ut}$). Without loss
  of generality, suppose that $w$ is in $P_{sv}$. Then let $P_{sw}$ be
  the subpath of $P_{sv}$ from $s$ to $w$, and $P_{wu}$ the subpath of
  $P_{sut}$ from $w$ to $u$. The path $P_{sw}\cup P_{wu}\cup (u,v)\cup
  P_{vt}$ goes from $s$ to $t$ and contains $(u,v)$.
\end{proof}

In a biconnected series-parallel graph $G$, a pair of vertices $(s,t)$
is a \emph{split pair} if the graph $G$ remains series-parallel after
adding an edge from $s$ to $t$.
% \footnote{It can be proven that $(s,t)$ is a $2$-vertex-cut, or
%   there is already an $(s,t)$ edge.}
In an oriented graph, a \emph{source} is a vertex that has only
outgoing edges, and a \emph{sink} is a vertex that has only incoming
edges.

\begin{lemma}\label{lem:orient}
  In a biconnected series-parallel graph $G$, for any split pair
  $(s,t)$, there is a unique way of orienting the edges of $G$ such
  that $G$ is acyclic, and $s$ and $t$ are the unique source and sink
  respectively. This orientation has the property that any simple path
  from $s$ to $t$ is oriented, and any oriented path can be extended
  into an oriented path from $s$ to $t$.
\end{lemma}
\begin{proof}
  Let $(u,v)$ be any edge in $G$. By Lemma~\ref{lem:suvt}, there is at
  least one simple path from $s$ to $t$ containing $(u,v)$. Suppose
  there are two such paths $P_1$ and $P_2$, connecting $s$--$u$--$v$--$t$ and
  $s$--$v$--$u$--$t$ in these orders respectively. These two paths plus
  an $(s,t)$ edge create a $K_4$ minor, which contradicts the fact that
  $(s,t)$ is a split pair. Therefore, there are either only paths
  connecting $s$--$u$--$v$--$t$ in that order, or only paths connecting
  $s$--$v$--$u$--$t$ in that order. We orient the edge $(u,v)$ in the order
  given by these paths. Trivially, any path from $s$ to $t$ is
  oriented. 

  We claim that orienting all edges in this way creates an acyclic
  orientation such that $s$ and $t$ are the unique source and sink
  respectively. Suppose that some vertex $u \ne s$ is a source. For any
  edge $(u,v)$, there is a simple path connecting $s$--$u$--$v$--$t$ in
  that order. Therefore this path is oriented, and so $u$ is not a
  source, a contradiction. Thus, $s$ is indeed the unique source.
  Symmetrically, $t$ is the unique sink. Suppose that the orientation
  creates an oriented cycle. For any edge $(u,v)$ in the cycle, there is
  a simple path $P$ connecting $s$--$u$--$v$--$t$ in that order. Let $w$
  and $z$ be the first and last vertex of the cycle in $P$. The cycle
  creates two paths from $w$ to $z$, one whose orientation must be
  inconsistent with the path connecting $s$--$w$--$z$--$t$ in that
  order; and so there are no oriented cycles. 
  
  Finally, for any oriented path, it is possible to extend it into an
  oriented path from $s$ to $t$ by adding edges at the beginning until
  it starts from $s$, and at the end until it ends at $t$.
\end{proof}

For an orientation of $G$ defined by a split pair $(s,t)$, if there is
an oriented path from $u$ to $v$, then $(u,v)$ is
\emph{compliant}. For any non-compliant pair of vertices $(u,v)$, let
$P_{sut}$ and $P_{svt}$ be two oriented paths from $s$ to $t$
containing $u$ and $v$ respectively. Let $P_{su}$, $P_{sv}$, $P_{ut}$
and $P_{vt}$ be the subpaths of $P_{sut}$ and $P_{svt}$ connecting the
two corresponding vertices. The pair $(w,z)$ is called the
\emph{terminals of $(u,v)$} if $w$ is the last common vertex of
$P_{su}$ and $P_{sv}$, and $z$ is the first common vertex of $P_{ut}$
and $P_{vt}$.
% The set of vertices that are on an oriented path from $w$ to $z$ is
% called the \emph{component containing $u$ and $v$}, or alternatively
% the \emph{component defined by $w$ and $z$}.
We prove now that the pair $(w,z)$ is independent of the choice of
$P_{sut}$ and $P_{svt}$. We say a pair of vertices $(w,z)$
\emph{separates} vertices $u$ from $v$ if $u$ and $v$ are in different
connected components of $V\setminus\{w,z\}$.
\begin{lemma}\label{lem:terminals}
  For any non-compliant pair $(u,v)$, there is a unique pair $(w,z)$
  of terminals of $(u,v)$. The pair $(w,z)$ is a $2$-vertex-cut
  separating $u$ from $v$. Furthermore, unless $s$ is $w$, $(w,z)$
  separates $u$ and $v$ from $s$, and unless $t$ is $z$, $(w,z)$
  separates $u$ and $v$ from $t$. Any simple cycle containing $u$ and
  $v$ also contains $w$ and $z$, is composed of two oriented paths
  from $w$ to $z$, and has $w$ as unique source and $z$ as unique
  sink.
\end{lemma}
\begin{proof}
  By Lemma~\ref{lem:sut}, there are simple paths $P_{sut}$ and
  $P_{svt}$ from $s$ to $t$ containing $u$ and $v$ respectively; by
  Lemma~\ref{lem:orient}, these paths are oriented. So $(u,v)$ always
  has at least one pair $(w,z)$ of terminals. Since $(s,t)$ is a split
  pair, we can assume there is an $(s,t)$ edge and still have $G$
  series-parallel. Then there are at least three vertex-disjoint paths
  from $w$ to $z$, one through $u$, one through $v$, and one
  containing $(s,t)$. So any path connecting vertices from two of
  these three paths must contain $w$ or $z$, because otherwise the
  graph would contain a $K_4$ minor. This means that $(w,z)$ is a
  $2$-vertex-cut separating $u$ from $v$; and if $s$ is not $w$ or $t$
  is not $z$, then they are also separated from $u$ and $v$ by
  $(w,z)$. This is true for any pair of terminals of $(u,v)$.

  Let $C$ be any simple cycle containing $u$ and $v$. Since $(w,z)$
  separates $u$ from $v$, $C$ must contain $w$ and $z$. Since a simple
  cycle can intersect only two connected components of
  $G\setminus\{w,z\}$, $C$ does not contain $s$ or $t$, unless they
  are $w$ or $z$ respectively. So $C$ is composed of two oriented
  paths from $w$ to $z$, containing $u$ and $v$ respectively. And so
  $w$ and $z$ are the unique source and sink of $C$.

  Since any simple cycle containing $u$ and $v$ also contains all the
  pairs of terminals of $(u,v)$, and has any pair of terminals as
  unique source and unique sink, there is only one pair of terminals
  of $(u,v)$.
%   Let $(w,z)$ and $(w',z')$ be two pairs of terminals of $(u,v)$.  Any
%   simple cycle containing $u$ and $v$ also contains $w$, $z$, $w'$ and
%   $z'$. Since $w$ and $w'$ both are the unique source of $C$, and $z$
%   and $z'$ both are its unique sink, then $(w,z)=(w',z')$ and $(u,v)$
%   has only one pair of terminals.
\end{proof}
% PROPERTY STILL TO BE WRITTEN: for two pairs of terminals $(w,z)$ and
% $(w',z')$ of different non-compliant demands, there can be no oriented
% path connecting $w-w'-z-z'$ in that strict order.

For any orientation of $G$ defined by a split pair $(s,t)$, a pair of
vertices $(w,z)$ is said to \emph{bracket} another pair $(u,v)$ if
there is an oriented path from $w$ to $z$ containing $u$ and $v$. If
$(w,z)$ brackets $(u,v)$ but $w\ne u$ or $z\ne v$, then $(w,z)$
\emph{strictly brackets} $(u,v)$. Since the orientation is acyclic,
the bracketing relation is transitive.

\begin{lemma}
\label{lem:sepa-embed}
Let $G=(V,E)$ be a series-parallel graph, and let $u,v\in V$ be two
arbitrary vertices. Then $G$ can be embedded in the plane so that 
both $u$ and $v$ are on the outside face.
\end{lemma}
\begin{proof}
If $G$ is series-parallel, then it does not contain
a $K_4$ minor.  Hence adding any single edge $e$ to $G$ does not create
either a $K_5$ or a $K_{3,3}$ minor; it follows that adding $e$ to $G$
results in a planar graph. In particular, $G' = (V, E\cup\{(u,v)\})$ is
planar. We embed $G'$ in the plane so that $(u,v)$ is on
the outside face.  (See, e.g.,~\cite{Schnyder}.)  Removing $(u,v)$ from
the result gives an embedding of $G$ with $u$ and $v$ on the outside face.
\end{proof}

% ---> Flow-cut gap and embedding distortion; bubbles introduced <---

\section{Congestion-Distortion Equivalence via LP Duality and Consequences}\label{sec:new}

We now give our new proof of the congestion-distortion equivalence
theorem (see Section~\ref{sec:intro}), using only basic notions of
linear programming duality. Our proof will reveal several additional
relations between LP variables that are useful later: in particular,
they give us cut-sufficiency conditions based on certain LP variables.
The starting point of the proof is a well-known fact: multiflows are
tightly related to metrics, because the dual of the LP expressing a
multiflow problem can be interpreted as the problem of finding a certain
graph metric.

\subsection{The Proof via LP Duality}

Throughout this section, we fix a ``supply graph'' $G = (V,E)$ and a
``demand graph'' $H = (V,F)$. The crux of the proof is to identify a
certain nonlinear maximization problem~\eqref{eqMaster} in variables $c
= \{c_e\}_{e\in E}$, $D = \{D_i\}_{i\in F}$, $l = \{l_e\}_{e\in E}$, and
$d = \{d_i\}_{i\in F}$ that has the following two properties. First, for
each setting of $c$ and $D$ satisfying the cut condition, the
program~\eqref{eqMaster} reduces to a maximization LP whose dual is the
problem of finding the minimum congestion for the multiflow problem
$(G,H,c,D)$. Second, for each setting of $l$ and $d$ satisfying certain
metric inequalities, the program~\eqref{eqMaster} reduces to a different
maximization LP whose dual is (a generalization of) the problem of
finding the minimum distortion embedding, into the family of cut-cone
metrics, of the metric given by $l$ and $d$. It follows that the maximum
possible congestion over all capacity/demand settings equals the maximum
possible distortion over all length settings. We now give the details.

For each $i\in F$, let $\{P_1^i,P_2^i,\ldots\}$ be a listing of
$\mathcal{P}[i]$, the set of simple paths in $G$ connecting the
endpoints of the demand $i$. The problem of determining the minimum
congestion for the multiflow instance $(G,H,c,D)$ can be written as%
\footnote{In this section, boldface is used to distinguish variables
from parameters in the linear programs.}
\begin{equation*}\tag{P}\label{eqP}
  \begin{array}{rrrrcll}
  z(c,D) & = & \min & \testbf{\alpha} &&&\\
  &&s.t. & \displaystyle\sum_{i,j:\,e\in P_j^i} \testbf{f_j^i} &\leq& c_e\testbf{\alpha} &\forall e\in E\\
  &&   &  \displaystyle \sum_j\testbf{f_j^i} & \geq & D_i & \forall i\in F\\
  &&   &  \testbf{f_j^i} & \geq & 0 & \forall i,j,
  \end{array}
\end{equation*}
where $\testbf{f_j^i}$ is the variable indicating the amount routed on path
$P_j^i$. The dual linear program is the following:
\begin{equation*}\tag{D}\label{eqD}
  \begin{array}{rrrrcll}
  z(c,D) & = & \max & \displaystyle \sum_iD_i\testbf{d_i} &&&\\
  &&s.t. & \displaystyle \sum_ec_e\testbf{l_e} & = & 1 &\\
  &&     & \testbf{d_i} &\leq &\displaystyle \sum_{e\in P_j^i}\testbf{l_e} &\forall i,j\\
  &&     & \testbf{d_i}         & \geq & 0 & \forall i\in F\\
  &&     &        \testbf{l_e}  & \geq & 0 & \forall e\in E.
  \end{array}
\end{equation*}
The variables $\testbf{l_e}$ can be thought of as \emph{lengths} of the
edges of $G$, and $\testbf{d_i}$ as \emph{distances} between the endpoints 
of $i$. The second set of constraints are \emph{metric inequalities},
which ensure that $\testbf{d_i}$ is no more than the shortest-path
distance between the endpoints of $i$ induced by the lengths
$\testbf{l_e}$.

% following text omitted by LF becasue does not seem relevant to rest of paper.
%
% There is a natural physical interpretation to the dual. If edge $e$ is
% considered as a pipe, with $c_e$ its cross-section, and $\testbf{l_e}$
% its length, then $c_e\testbf{l_e}$ is the volume of the pipe. The
% first constraint then expresses that the total volume of the pipe
% network is fixed to one. The objective function expresses the total
% volume necessary to route all demands in the network, each demand
% requiring a total cross-section of $D_i$, which must be available on a
% length equal to the shortest path routing the demand. Linear
% programming duality tells us that solutions of the dual lower bound to
% those of the primal. Intuitively, if the total volume available is
% $1$, and the quantity of volume necessary to route the demand is
% expressed by the objective function of~\eqref{eqD}, then that value is
% a lower bound on the congestion required, since the congestion
% multiplies the quantity of volume available in the graph.

In order to find the flow-cut gap of a pair $(G,H)$, we need to find the
maximum value to~\eqref{eqP} (and~\eqref{eqD}) over all choices of
capacities $c$ and demands $D$ that satisfy the cut condition, which can
be expressed as
\begin{equation}\label{eqMast1}
  \begin{array}{rrcll}
  \max & z(\testbf{c},\testbf{D}) &&&\\
  s.t. &\displaystyle \sum_{i\in\delta_H(C)} \testbf{D_i} &\leq& \displaystyle\sum_{e\in\delta_G(C)}\testbf{c_e} &\forall C\subseteq V\\
       & \testbf{D_i}         & \geq & 0 & \forall i\in F,\\
       &        \testbf{c_e}  & \geq & 0 & \forall e\in E.
  \end{array}
\end{equation}
Since the linear program~\eqref{eqD} is a maximization problem, we can
write the problem of finding the flow-cut gap as a single maximization
problem on variables $\testbf{c}$, $\testbf{D}$, $\testbf{l}$ and
$\testbf{d}$:
\begin{equation*}\tag{+}\label{eqMaster}
  \begin{array}{rrcll}
    \max & \displaystyle \sum_i\testbf{D_id_i} &&&\\
   s.t. &  \displaystyle \sum_e\testbf{c_el_e} & = & 1 &\\
   & \testbf{d_i} &\leq &\displaystyle \sum_{e\in P_j^i}\testbf{l_e} &\forall i,j\\
   & \displaystyle\sum_{i\in\delta_H(C)} \testbf{D_i} &\leq& \displaystyle\sum_{e\in\delta_G(C)}\testbf{c_e} &\forall C\subseteq V\\
    & \testbf{D_i},\testbf{d_i}          & \geq & 0 & \forall i\in F\\
    &         \testbf{c_e},\testbf{l_e}  & \geq & 0 & \forall e\in E.
  \end{array}
\end{equation*}
This is not a linear program, since some of the variables
multiply each other. However, there are two ways we can transform it
into a linear program by setting some variables to be parameters. If
we fix $\testbf{c_e}$ for all $e$ and $\testbf{D_i}$ for all $i$ to be
parameters that satisfy the cut condition, we obtain the linear
program~\eqref{eqD}. But if we fix $\testbf{l_e}$ for all $e$ and
$\testbf{d_i}$ for all $i$ to be parameters that satisfy the metric
inequalities, we find a different linear program in variables
$\testbf{c_e}$ and $\testbf{D_i}$:
\begin{equation*}\tag{D$'$}\label{eqD'}
  \begin{array}{rrrrcll}
  w(l,d) & = & \max & \displaystyle \sum_id_i\testbf{D_i} &&&\\
  &&  s.t. & \displaystyle \sum_el_e\testbf{c_e} & = & 1 &\\
  &&       & \displaystyle\sum_{i\in\delta_H(C)} \testbf{D_i} &\leq& \displaystyle\sum_{e\in\delta_G(C)}\testbf{c_e} &\forall C\subseteq V\\
  &&  & \testbf{D_i}          & \geq & 0 & \forall i\in F\\
  &&  &         \testbf{c_e}  & \geq & 0 & \forall e\in E.
  \end{array}
\end{equation*}
The flow-cut gap problem (\ref{eqMast1}) and (\ref{eqMaster}) 
can then also be expressed as
\begin{equation}\label{eqMast2}
  \begin{array}{rrrcll}
  \max & w(\testbf{l},\testbf{d}) &&&\\
  s.t. & \testbf{d_i} &\leq &\displaystyle \sum_{e\in P_j^i}\testbf{l_e} &\forall i,j\\
       & \testbf{d_i}          & \geq & 0 & \forall i\in F\\
       &         \testbf{l_e}  & \geq & 0 & \forall e\in E.
  \end{array}
\end{equation}
Notice that in a solution achieving the maximum above, each
$\testbf{d_i}$ must equal the shortest-path distance between the
endpoints of $i$ induced by the lengths $\testbf{l_e}$.  The dual
of~\eqref{eqD'} is
\begin{equation*}\tag{P$'$}\label{eqP'}
\begin{array}{rrrrcll}
  w(l,d) &=& \min & \testbf{\gamma} &&&\\
  &&s.t. & \displaystyle\sum_{C:e\in \delta_G(C)} \testbf{x_C} &\leq& l_e\gamma &\forall e\in E,\\
  &&     & \displaystyle \sum_{C:i\in\delta_H(C)} \testbf{x_C} &\geq & d_i & \forall i\in F\\
  &&     &         \testbf{x_C} & \geq & 0 & \forall C\subseteq V.
\end{array}
\end{equation*}
The system~\eqref{eqP'} has a variable $\testbf{x_C}$ for each subset $C
\subseteq V$.  The values of these variables define a cut-cone metric;
call it $d'$. The first constraint says that the $d'$-length of an edge
$e$ is at most $\testbf{\gamma}$ times its ``true'' length $l_e$.  The
second constraint says that the $d'$-distance between the endpoints of a
demand $i$ is at least $d_i$, which, for $l$ and $d$ achieving the
maximum in~\eqref{eqMast2}, equals the ``true'' distance given by $l$.
Thus, (\ref{eqP'}) can be seen as approximating (at least between
endpoints of demands) the shortest-path metric induced by $l$ by a
cut-cone metric, within an approximation factor $\testbf{\gamma}$ as
small as possible.

As a clean special case, when $H$ is a complete graph on $V$, then
$d_e=l_e$ for each edge $e$, and thus the two constraints
in~\eqref{eqP'} say (respectively) that $d \ge d'/\gamma$ and that the
distortion from $d$ to $d'/\gamma$ is at most $\gamma$.  Thus, $d$
embeds into the family of cut-cone metrics with distortion at most
$\gamma$. The equivalence of (\ref{eqMast1}), (\ref{eqMaster}), and
(\ref{eqMast2}) means that the flow-cut gap of $(G,H)$ is equal to the
minimum distortion required to embed an arbitrary shortest-path metric
defined on $G$ into the family of cut-cone metrics. This completes the
proof of the congestion-distortion equivalence, entirely through basic
notions of linear programming.

% Omitted by LF since there doesn't seem to be a purpose to these
% for the rest of the paper
%
% We can see in the linear program~\eqref{eqP'} the precise constraints
% on the approximation of the shortest-path distance defined by $l$ and
% $d$ by a cut metric. In particular, we can see that there are no
% constraints on distances between vertices that are not connected by a
% supply or a demand edge.  We can also see that it is unnecessary to
% check that the distances approximating $d$ in the cut metric are small
% enough, or that the distances approximating the lengths $l$ are large
% enough.

\subsection{Implications of the New Proof}

Suppose that, for some pair of graphs $(G,H)$, with $G = (V,E)$, $H =
(V,F)$, we have an optimal solution $(c^*,D^*,l^*,d^*)$ to the nonlinear
program~\eqref{eqMaster}. By the properties of linear programming
duality, there are solutions $f^*$ and $x^*$ to the flow
problem~\eqref{eqP} and the cut metric problem \eqref{eqP'} that satisfy
complementary slackness. We call $(c^*,D^*,l^*,d^*,f^*,x^*)$ a
\emph{general solution} to the pair $(G,H)$.

\begin{lemma}
\label{lem:cs}
A general solution satisfies the following properties. \\
$\langle a \rangle$ If $x_C^*>0$, then 
$\sum_{i\in\delta_H(C)} D_i^* = \sum_{e\in\delta_G(C)}c_e^*$.  Thus, only
tight cuts have positive $x$-value in (\ref{eqP'}). \\
$\langle b \rangle$ If $f^{i*}_j > 0$ then $d^*_i = \sum_{e\in P^i_j} l^*_e$;
and for each $i$ with $D^*_i > 0$, there is a path $P^i_{j*}$
for which this is true. 

\end{lemma} 
\begin{proof}
$\langle a \rangle$ follows from complementary
  slackness applied to (\ref{eqP'}) and (\ref{eqD'}).
  The first part of $\langle b \rangle$ follows from complementary
  slackness applied to (\ref{eqP}) and (\ref{eqD}). The second part
  of $\langle b \rangle$ follows from the second constraint of
  (\ref{eqP}).  
\end{proof}

\begin{lemma}\label{lem:tight}
  There is a solution $x^*$ to the problem~\eqref{eqP'} such that only
  central cuts have a positive $x$-value.
\end{lemma}
\begin{proof}
  Suppose that in the optimal solution, $x_C^*>0$ for some $C$ 
that can be decomposed into two sets $C_1$ and $C_2$
  that are not connected by any supply edge.  Adding the value of
$x_C^*$ to the values of $x_{C_1}^*$ and $x_{C_2}^*$ and setting
$x_C^*$ to zero increases the distance, in the
  cut metric defined by $x^*$, between all pairs of vertices $u\in
  C_1$, $v\in C_2$.  Since there is no supply edge from $u$ to $v$,
  there is no upper constraint on this distance in the linear
  program~\eqref{eqP'}; and so the new solution is still optimal. By
  induction, there is an optimal solution with $x_C^*=0$ for any
  non-central $C$.
\end{proof}

\noindent
We assume from now on that in a general solution, if $x_C^*>0$ then
$C$ is central.\footnote{Another way to do this is to decide from the
beginning that the optimization program~\eqref{eqMast1} only has cut
condition constraints on central cuts, as this is sufficient for
ensuring the cut condition is satisfied, which implies that the linear
program~\eqref{eqP'} only has variables $x_C$ for central cuts.}

We define a \emph{simple pair} to be a pair $(G,H)$ that has a general
solution such that $c_e^*>0$ and $l_e^* > 0$ for each $e$, and $D_i^*>0$
for each $i$.

\begin{lemma}\label{lem:simple}
For each pair $(G,H)$, there is a simple pair with the same flow-cut gap. 
\end{lemma}
\begin{proof}
  Suppose that some edge $e$ has a zero capacity in the optimal
  solution to~\eqref{eqMaster} (i.e. $c_e^*=0$). This means that there
  is no upper constraint on the value of $l_e$, and so the constraints
  for paths $P_j^i$ which contain edge $e$ put no restriction on the
  value of $d_i$; the constraints on cuts containing $e$ do not change
  if $c_e$ is in the expression and equal to zero, or removed from the
  expression. Thus, deleting the edge $e$ from $G$ does not change
  constraints on $d_i$ and $D_i$, and so the flow-cut gap remains the
  same. Similarly, if $l_e^*=0$, there is no upper constraint on the
  value of $c_e$, and so the constraints for cuts $C$ which contain
  edge $e$ put no restriction on the value of demands crossing $C$;
  the constraints on paths containing $e$ do not change if $l_e$ is in
  the expression and equal to zero, or removed from the expression.
  Thus, contracting the edge $e$ does not change constraints on $d_i$
  and $D_i$, and so the flow-cut gap remains the same. If $D_i^*=0$
  for some $i$, it makes no difference what constraints are on $d_i$,
  and so the flow-cut gap remains the same if the demand $i$ is
  deleted from $H$.
\end{proof}

We assume from now on that $(G,H)$ is a simple pair. In what follows,
recall that $\mathcal{P}[i]$ denotes the set of paths in $G$ that
connect the endpoints of the demand edge $i\in F$.  

\begin{lemma}\label{lem:main}
  Let $(G,H)$ be a simple pair with general solution 
  $(c^*,D^*,l^*,d^*,f^*,x^*)$. Then
  \begin{equation} \label{eqgamma}
    \forall\,i\in F~
    \forall\,P\in \mathcal{P}[i]\,:~~
    \sum_{C:i\in\delta_H(C)}x_C^*
    =   d_i^*
    \le \sum_{e\in P}l_e^* 
    =   \frac{1}{\gamma^*}\sum_{e\in P}\sum_{C:e\in \delta_G(C)} x_C^*,
  \end{equation}
with equality if $P$ is a shortest path for the shortest-path metric
defined by $l^*$, e.g., if the solution has a nonzero flow routing the
demand $i$ along $P$.
\end{lemma}
\begin{proof}
If $(c^*,D^*,l^*,d^*,f^*,x^*)$ is a general solution for the simple pair
$(G,H)$, then by definition we have $c_e^* > 0$ for all $e$ and $D_i^* >
0$ for all $i$. By complementary slackness applied
to (\ref{eqP'}) and (\ref{eqD'}), this implies that the inequalities 
of~\eqref{eqP'} are all tight:
\begin{equation*}
% \label{eq:eTight}
  l_e^*\gamma = \sum_{C:e\in \delta_G(C)} x_C^*\quad\forall e\in E,
\end{equation*}
\begin{equation*}
% \label{eq:iTight}
  d_i^* = \sum_{C:i\in\delta_H(C)} x_C^*\quad\forall i\in F.
\end{equation*}
Notably, this implies that the length of any edge $e$ in the cut-cone
metric defined by $x^*$ is always $\gamma^*$ times $l_e^*$. 
These equalities imply the equalities in (\ref{eqgamma}).
The inequality holds for all solutions. 
Lemma~\ref{lem:cs}$\langle b \rangle$ shows 
that the inequality is tight for at least one path, and hence it is
tight for the shortest path.
\end{proof}

And so, in the metric defined by $x^*$, for any path $P_j^i$ such
that $f_j^{i*}>0$, the ratio between the sum of lengths of edges in
the path and the distance between endpoints of the path is equal to
the flow-cut gap.

\begin{theorem}\label{thm:bubble}
  Let $(G,H)$ be a simple pair with general solution
  $(c^*,D^*,l^*,d^*,f^*,x^*)$. Suppose there exist $i\in F$ and
  $P\in\mathcal{P}[i]$ such that $P$ crosses each tight cut at most
  once, with tightness defined according to $c^*$ and $D^*$. Then
  $(G,H)$ is cut-sufficient, i.e., its flow-cut gap is one.  More
  explicitly, a multiflow problem on $(G,H)$ has a fractional solution
  for any choice of capacities and demands that satisfy the cut
  condition.
\end{theorem}
\begin{proof}
  Pick a tight set $C\subseteq V$. The number of times that $P$ crosses
  $C$ is odd if $i\in\delta_H(C)$ and even otherwise. Therefore, by the
  given condition, if $i\notin\delta_H(C)$, then $P$ must not cross $C$;
  otherwise $P$ must cross $C$ exactly once.  Recall that, by
  Lemma~\ref{lem:cs}$\langle a \rangle$, only tight cuts may have
  non-zero $x^*$-values. This implies that
  \[
  \sum_{C:i\in\delta_H(C)}x^*_C =\sum_{e\in P}\sum_{C:e\in \delta_G(C)} x^*_C.
  \]
  In view of Lemma~\ref{lem:main}, this means that
  \[
  \gamma\sum_{C:i\in\delta_H(C)}x^*_C \leq 
  \sum_{e\in P}\sum_{C:e\in \delta_G(C)} x^*_C =
  \sum_{C:i\in\delta_H(C)} x^*_C.
  \]
  Since the pair $(G,H)$ is
  simple, each $l^*_e$ is non-zero and thus, so is
  $\sum_{C:i\in\delta_H(C)} x^*_C$.  It follows that $\gamma=1$ and the
  flow-cut gap is one, as claimed.
\end{proof}

% ---> The main structure theorem <---

% \section{Proof of Theorem~\ref{thm:fraction}}\label{sec:fraction}
\section{Proof of the Fractional Routing Theorem} \label{sec:fraction}

In this section, we prove Theorem~\ref{thm:fraction}. Namely, for
series-parallel graphs $G$, we show that the pair $(G,H)$ is
cut-sufficient if and only if it does not contain an odd spindle
as a minor. The following special case of this theorem was proven
earlier in Chekuri et al.~\cite{Chekuri10}, and we use it in our proof.

\begin{theorem}[{\cite[Section 3.3]{Chekuri10}}] \label{thm:k2m}
  Suppose $G$ is $K_{2,m}$, with possibly an additional supply edge
  between the two vertices not of degree $2$. Then $(G,H)$ is
  cut-sufficient iff it does not contain an odd spindle as a
  minor.
\end{theorem}

The ``only if'' direction of Theorem~\ref{thm:fraction} is easy, and is
proven in Section 3.3 of~\cite{Chekuri10}; we reproduce the argument
here.  An odd spindle itself has a flow-cut gap of more than $1$,
as can be seen by setting the capacity of all supply edges and the
demand of all demand edges to $1$.  Let a pair $(G,H)$ contain a pair
$(G',H')$ as a minor, and let $(G',H',c',D')$ be an instance of the
multiflow problem. We assign capacities $c$ and demands $D$ to the pair
$(G,H)$ in the following way. To any supply edge or demand edge that is
deleted during the reduction from $(G,H)$ to $(G',H')$, we assign a
capacity or demand of $0$. To any supply edge that is contracted, we
assign a very large capacity. And to any edge of $(G,H)$ that is still
in $(G',H')$ after the reduction, we assign the capacity or demand of
the corresponding edge in $(G',H')$. Since $(G',H',c',D')$ satisfies the
cut condition, so does $(G,H,c,D)$.  For any multiflow solving the
instance $(G,H,c,D)$ with congestion $\gamma$, we build a multiflow
solving the instance $(G',H',c',D')$ with the same congestion $\gamma$,
by sending on each edge of $G'$ the same flow as on the corresponding
edge in $G$. Therefore, the minimum congestion for $(G,H)$ cannot be
less than the minimum congestion for $(G',H')$. And so, a pair $(G,H)$
cannot be cut-sufficient if it has as a minor a pair $(G',H')$ that is
not cut-sufficient.

We now prove the ``if'' direction. Suppose the pair $(G,H)$ has flow-cut
gap more than $1$. By Lemma~\ref{lem:simple}, we may
assume that $(G,H)$ is simple. 
For a demand $(u,v)$, a \emph{bubble} for $(u,v)$ is a central set defining a tight
cut, but containing neither $u$ nor $v$. The set $\mathcal{P}[u,v]$ (of
paths in $G$ between $u$ and $v$) is {\em covered by bubbles} if every
path in it crosses a bubble at least once. From
Theorem~\ref{thm:bubble}, and a parity argument, we get the
following:
\begin{observation}\label{obs:odd-cross}
  If $P\in\mathcal{P}[u,v]$ does not cross any bubble, then $P$
  crosses some tight cut $t > 1$ times, where $t$ is odd.
  \qed
\end{observation}
To prove Theorem~\ref{thm:fraction}, we first prove that if there is a
demand $(u,v)$ such that $\mathcal{P}[u,v]$ is covered by
bubbles, then the instance must contain an odd spindle as a minor
(Lemma~\ref{lem:bubbles}). We then prove that there must be such a
demand (Lemma~\ref{lem:existbubbles}).

% AC: rewrote the material below (see above), pulling 
% out the odd-crossing fact as an "Observation".
%
% We now turn to the ``if'' direction. By Lemma~\ref{lem:simple}, we can
% assume the pair $(G,H)$ is simple. By Theorem~\ref{thm:bubble}, if a
% simple instance has a flow-cut gap of more than $1$, then any path
% routing a demand crosses some tight cut more than once. For any demand
% $(u,v)$, a \emph{bubble} for $(u,v)$ is a central set defining a tight
% cut, but containing neither $u$ nor $v$. The tight cuts defined by
% bubbles are crossed an even number of times by any path from $u$ to
% $v$. If an instance has a flow-cut gap of more than $1$, then for
% every demand $(u,v)$, all simple paths from $u$ to $v$ are either
% covered by bubbles, or cross some tight cut a number of times that is
% odd and more than $1$. We first prove that if there is a demand
% $(u,v)$ such that all simple paths from $u$ to $v$ are covered by
% bubbles, then the instance must contain an odd spindle as a minor
% (Lemma~\ref{lem:bubbles}). We then prove that there must be such a
% demand (Lemma~\ref{lem:existbubbles}).

\subsection{Coverage by bubbles creates an odd spindle minor}

\begin{lemma}
\label{lem:bubbles}
If there is a demand $(u,v)$ such that $\mathcal{P}[u,v]$ is
covered by bubbles, then the instance must contain an 
odd spindle as a minor.
\end{lemma}
\begin{proof}
Let $F_{u,v}$ be a minimal family of bubbles covering all simple paths
from $u$ to $v$. We first claim that $|F_{u,v}| \ge 2$. Indeed, if
$F_{u,v} = \{B\}$ for a bubble $B$, then the vertices $u$ and $v$ are in
different connected components of $V\setminus B$. This contradicts the
fact that $B$ is central.  We now distinguish the following two cases:
$|F_{u,v}|\geq 3$, or $|F_{u,v}|=2$. The proof for each case proceeds
using a sequence of claims.
\bigskip

\noindent\textsc{Case 1 of Lemma~\ref{lem:bubbles}:} $|F_{u,v}| \ge 3$.

\begin{claim} If $|F_{u,v}|\geq 3$, then the bubbles in $F_{u,v}$
  are disjoint, and there is no edge in $G$ going from one bubble to 
  another.
\end{claim}
\begin{proof}
  For each bubble in $F_{u,v}$, there is a path crossing it that does
  not cross any other bubble in $F_{u,v}$ (otherwise we could remove
  that bubble and $F_{u,v}$ would not be minimal). Suppose bubbles $A$
  and $B$ intersect. Let $P_A$, $P_B$ be the paths through $A$ and $B$
  respectively. Consider $p\in A\cap B$. Since $A$ and $B$ are both
  connected, there is a path in $A$ from a node in $P_A\cap A$ to $p$
  and a path in $B$ from $P_B \cap B$ to $p$. This creates a $K_4$ minor
  with any third path from $u$ to $v$, which exists since
  $|F_{u,v}|\geq 3$, contradicting the series-parallelness of $G$; and so
  $A$ and $B$ do not intersect. If there is an edge connecting the
  bubbles $A$ and $B$, there is again a path in $A\cup B$ connecting
  $P_A$ to $P_B$, which again creates a $K_4$ minor.
\end{proof}

We contract every edge that does not cross one of the cuts defined by
the bubbles in $F_{u,v}$. We get one vertex $f_i$ for each bubble, one
vertex $u'$ for the part of the graph reachable from $u$ without
crossing the bubbles, and one vertex $v'$ for the part reachable from
$v$ without crossing the bubbles.  We prove there are no other
vertices.

\begin{claim}\label{clm:contract-to-k2m}
  The contracted supply graph is a $K_{2,m}$.
\end{claim}
\begin{proof}
  We know that in the uncontracted graph, there is a path connecting
  $u$ to $v$ through each bubble, disjoint from the other bubbles.
  And so, there is an edge from $u'$ and $v'$ to each $f_i$, and these
  vertices induce a $K_{2,m}$ subgraph. Suppose there is another vertex
  $x$. The vertex $x$ cannot be adjacent to $u'$ or $v'$, because the
  edge between them would have been contracted. It cannot be connected
  to two different vertices $f_i$ and $f_j$, because this would create
  a $K_4$. So it is connected to a single $f_i$, and it is a leaf. But
  the set $\{f_i\}$ defines a tight cut, and since $x$ is a leaf, the
  set $\{f_i,x\}$ would define a cut with a smaller surplus than
  $\{f_i\}$, which is not possible. So $x$ does not exist.
\end{proof}

The contracted instance has a $K_{2,m}$ supply graph. Each vertex $f_i$
of degree $2$ defines a tight cut, since it is the result of contracting
a tight set. So in any fractional solution to the contracted instance,
the two supply edges leaving $f_i$ have just enough capacity to route
the demands incident to $f_i$, and no flow can go from $u'$ to $v'$
through $f_i$. Since there is a demand from $u'$ to $v'$, this means
that the instance does not have a solution, and therefore, by
Theorem~\ref{thm:k2m}, it contains an odd spindle as a minor.

This finishes the case $|F_{u,v}| \ge 3$.
\bigskip

\noindent\textsc{Case 2 of Lemma~\ref{lem:bubbles}:} $|F_{u,v}| = 2$.
\bigskip

Suppose $F_{u,v} = \{A,B\}$, for distinct bubbles $A$, $B$.  By a
sequence of claims, we prove that if we contract every edge that does
not cross a bubble in $F_{u,v}$, we get an instance with a $K_{2,m}$
supply graph, satisfying the cut condition, but unroutable. Appealing to
Theorem~\ref{thm:k2m} again, we conclude that the instance contains an
odd spindle as minor.

\begin{claim}\label{clm:2inter}
  If $A$ and $B$ are two bubbles covering every simple path from $u$
  to $v$, then $A$ and $B$ intersect.
\end{claim}
\begin{proof}
  Let $R$ be the connected component of $V\setminus(A\cup B)$
  containing $v$. Let $X=A \cup R$ and $Y=B \cup R$. 
  Suppose $A$ and $B$ are disjoint. Then $X\cap Y = R$.

  Now $\sigma(X\setminus Y) = \sigma(A) = 0$, and $\sigma(Y\setminus X) = \sigma(B) =
  0$. By Lemma~\ref{lem:easy1}$\langle d \rangle$, we have $\sigma(X)+\sigma(Y) = \sigma(X\setminus Y) +
  \sigma(Y\setminus X) + 2\sigma(X\cap Y,V\setminus (X\cup Y))<0$, because
  $\sigma(X\cap Y,V\setminus (X\cup Y))$ includes the demand $(u,v)$
  and $X\cap Y = R$ which is disconnected from the rest of the graph
  by $A$ and $B$. However, by the cut condition, $\sigma(X) \ge 0$ and
  $\sigma(Y) \ge 0$, a contradiction. Therefore $A$ and $B$ intersect.
\end{proof}

\begin{claim}\label{clm:2cut}
  There are two vertices, taken from $A\setminus B$ and
  $B\setminus A$ respectively, that form a $2$-vertex-cut of $G$,
  separating it into at least three connected components, with $u$ and
  $v$ in different components.
\end{claim}
\begin{proof}
  Let $U$ be the connected component (in $G$) of $V\setminus (A\cup B)$
  containing $u$, and let $R$ be the connected component of
  $V\setminus(A\cup B)$ containing $v$. Since $A$ is central, there is
  a path from $u$ to $v$ outside $A$ which goes through $B\setminus
  A$. Symmetrically, there is a path from $u$ to $v$ outside $B$ which
  goes through $A\setminus B$. These two paths form a cycle $C$ going
  through $U$, $A\setminus B$, $R$ and $B\setminus A$ in order. By
  Claim~\ref{clm:2inter}, there is a vertex $x\in A\cap B$. Since $x$
  is in $A$, there is a path $P_a$ in $A$ from $x$ to $C\cap
  (A\setminus B)$. Let $a$ be the endpoint of $P_a$ on $C$. Since $x$
  is in $B$, there is a path $P_b$ in $B$ from $x$ to $C\cap
  (B\setminus A)$. Let $b$ be the endpoint of $P_b$ on $C$. The paths
  $P_a$ and $P_b$ only intersect in $A\cap B$. So there are three
  vertex-disjoint paths in $G$ from $a$ to $b$, one through $U$, one 
  through $R$, and one through $A\cap B$. So $(a,b)$ must be a 
  $2$-vertex-cut, for otherwise $G$ would have a $K_4$ minor.
\end{proof}

\begin{claim}\label{clm:abcentral}
  The sets $A\setminus B$ and $B\setminus A$ are both central.
\end{claim}
\begin{proof}
  By Claim~\ref{clm:2cut}, $A\setminus B$ and $B\setminus A$ contain a
  pair of vertices that is a vertex $2$-cut separating $u$ from $v$.
  We use Lemma~2.4 of \cite{Chekuri09}, which proves that in a
  series-parallel graph, this implies that $A\setminus B$ and $B\setminus A$
  are both central.
\end{proof}

\begin{claim}\label{clm:contract-ab}
  If we contract every edge of $G$ that is neither in $\delta_G(A)$,
  nor in $\delta_G(B)$, and merge parallel edges, we get a $K_{2,m}$,
  with possibly one extra supply edge connecting the two vertices not
  of degree $2$.
\end{claim}
\clearpage
\begin{proof}%
\vspace{-1.15\baselineskip} % Undo the blank para inserted by \parpic
\parpic[r]{
  \begin{minipage}{2in}
  \begin{center}
\psset{unit=0.8cm,arrows=-,shortput=nab,linewidth=0.5pt,arrowsize=2pt 5,labelsep=1.5pt}
\pspicture(0,0.3)(4,4.3)
\psframe[framearc=.5](1.65,0)(2.35,1.7)
\psframe[framearc=.5](1.65,2.3)(2.35,4)
\pscircle(2,0.35){0.2}
\pscircle(2,0.85){0.2}
\pscircle(2,1.35){0.2}
\pscircle(2,2.65){0.2}
\pscircle(2,3.15){0.2}
\pscircle(2,3.65){0.2}
\pscircle(0,2){0.5}
\pscircle(4,2){0.5}
\dotnode[linecolor=white,linewidth=5pt](0,2){a}
\dotnode[linecolor=white,linewidth=5pt](4,2){b}
\dotnode[linecolor=white](2,0.35){k1}
\dotnode[linecolor=white](2,0.85){k2}
\dotnode[linecolor=white](2,1.35){k3}
\dotnode[linecolor=white](2,2.65){k4}
\dotnode[linecolor=white](2,3.15){k5}
\dotnode[linecolor=white](2,3.65){k6}
\ncline{a}{k1}
\ncline{a}{k2}
\ncline{a}{k3}
\ncline{a}{k4}
\ncline{a}{k5}
\ncline{a}{k6}
\ncline{b}{k1}
\ncline{b}{k2}
\ncline{b}{k3}
\ncline{b}{k4}
\ncline{b}{k5}
\ncline{b}{k6}
\ncline[linestyle=dashed]{a}{b}
\rput(0,2.75){$A\setminus B$}
\rput(4,2.75){$B\setminus A$}
\rput[lB](2.5,0){$A\cap B$}
\rput[lB](2.5,4){$V\setminus (A\cup B)$}
\endpspicture
\end{center}
  \end{minipage}
}%
  \hspace{3ex}
  Since $A\setminus B$ is central, it is connected. Similarly,
  $B\setminus A$ is connected. The rest of the graph is composed of
  $A\cap B$, which has at least one connected component by
  Claim~\ref{clm:2inter}, and $V\setminus(A\cup B)$, which has at
  least two connected components containing $u$ and $v$
  respectively. There is an edge connecting $A\setminus B$ to each
  connected component of $A\cap B$ (because both are in $A$, which is
  central), and there is an edge connecting $A\setminus B$ to each
  connected component of $V\setminus (A\cup B)$ (because neither is in
  $B$, which is central). Similarly, there is an edge connecting
  $B\setminus A$ to each connected component of $A\cap B$ and
  $V\setminus (A\cup B)$. This implies that for each connected
  component of $A\cap B$ and $V\setminus (A\cup B)$, there is a path
  connecting $A\setminus B$ to $B\setminus A$ through that
  component. As a consequence, there is never an edge going from a
  connected component of $A\cap B$ to a connected component of
  $V\setminus (A\cup B)$, because this would create a $K_4$ with
  $A\setminus B$ and $B\setminus A$, which are also connected through
  at least another connected component of $V\setminus (A\cup B)$.
  % The connectivity between $A\setminus B$, $B\setminus A$, and
  % connected components of $A\cap B$ and $V\setminus (A\cup B)$ is
  % not affected by the contraction.
\end{proof}

Let us perform the contraction described in Claim~\ref{clm:contract-ab}.
After the contraction, the endpoints of the demand edge $(u,v)$ are
still separated by the sets $A$ and $B$, which are still tight. So the
contracted instance is not routable, even though it satisfies the cut
condition. And so, the contracted pair of graphs is not
cut-sufficient. But since the contracted supply graph is a $K_{2,m}$,
Theorem~\ref{thm:k2m} implies that the contracted pair contains an odd
spindle as minor.  Therefore, so does the original pair $(G,H)$. We
are now done with the case $|F_{u,v}| = 2$.
\hfill\hfill\hfill\hfill\hfill\hfill\emph{This completes the proof of
  Lemma~\ref{lem:bubbles}.}
\end{proof}

\subsection{Identifying a bubble-covered demand}

To finish the proof of Theorem~\ref{thm:fraction}, we must show that the
conditions of Lemma~\ref{lem:bubbles} are satisfied, so that our
instance $(G,H)$ {\em does} have an odd spindle as a minor. The next
lemma shows precisely this. The proof of this lemma uses the notions of
{\em split pairs} and {\em bracketing}, defined in
Section~\ref{sec:graph-props}.

\begin{lemma}
\label{lem:existbubbles}
If a simple pair $(G,H)$ has flow-cut gap greater than $1$, then there
is a demand $(u,v)$ such that $\mathcal{P}[u,v]$ is covered by bubbles.
\end{lemma}
\begin{proof}
  We choose an arbitrary split pair in graph $G$, and orient $G$
  accordingly. By Theorem~3.1 of~\cite{Chekuri09}, there must be at
  least one non-compliant demand. We then choose a non-compliant
  demand $(u,v)$ such that its pair of terminals does not strictly
  bracket the pair of terminals of any other non-compliant
  demand. This is always possible, since in the set of pairs of
  terminals, the bracket relation is a partial order and must have a
  minimal pair.

  Suppose $\mathcal{P}[u,v]$ is {\em not} covered by bubbles. We shall
  demonstrate a contradiction with our choice of $(u,v)$.  By
  Observation~\ref{obs:odd-cross}, a path in $\mathcal{P}[u,v]$ not
  covered by a bubble must cross some tight cut an odd number of times,
  more than once.

% \begin{lemma}\label{lem:multicross}
%   In a series-parallel graph, let a path $P$ cross a central cut
%   $\delta_G(C)$ at least four times. Let $U_1$, $U_2$ and $U_3$ be the
%   first three connected components of $P\setminus \delta_G(C)$ on the
%   same side of $\delta_G(C)$, in the order they appear in $P$. Then
%   any path connecting $U_1$ to $U_3$ without crossing $\delta_G(C)$
%   must go through $U_2$.
% \end{lemma}
% \begin{proof}
%   Suppose there is a path connecting $U_1$ to $U_3$ without going
%   through $U_2$. Then that path, together with $P$, forms a simple 
%   cycle that crosses the central cut $C$ four times. By 
%   Lemma~\ref{lem:cycle}, this is a contradiction.
% \end{proof}

% Let $(u,v)$ be a demand, with a shortest path $P$ from $u$ to $v$ that
% is not covered by any bubble, and crosses a tight cut $\delta_G(C)$ an
% odd number of times, five or more. We contract all edges of $P$ that
% do not cross $\delta_G(C)$. Then by Lemma~\ref{lem:multicross}, the
% vertices of $P$ on each side of the cut $\delta_G(C)$ are connected by
% paths in the same order as $P$. By contracting these paths from the
% second connected component of $P$ on each side of the cut
% $\delta_G(C)$, we get an instance where $P$ crosses the cut
% $\delta_G(C)$ only three times; and so it is enough to prove that such
% an instance cannot exist.

\parpic[r]{
  \begin{minipage}{3.3cm}
\begin{center}
\psset{unit=0.8cm,arrows=-,shortput=nab,linewidth=0.5pt,arrowsize=2pt 5,labelsep=2pt}
\pspicture(-1,-0.1)(3.3,3.8)
\dotnode(0,3){u}
\dotnode(2,2.5){a}
\dotnode(0,0.5){b}
\dotnode(2,0){v}
\nput{-135}{u}{$u$}
\nput{45}{a}{$a$}
\nput{-135}{b}{$b$}
\nput{45}{v}{$v'$}
\ncline[linewidth=1pt]{u}{a}
\ncline[linewidth=1pt]{b}{a}
\ncline[linewidth=1pt]{b}{v}
\ncarc{b}{u}
\ncarc{a}{v}
\psline[linestyle=dotted,linewidth=1pt](1,0)(1,3.8)
\uput[180](1,3.8){$C$}
\psframe[linestyle=dotted,linewidth=1pt,framearc=.5](-1.2,3.2)(0.8,2.5)
\uput[0](-1.2,2.8){$S_1$}
\psframe[linestyle=dotted,linewidth=1pt,framearc=.5](-1,1.7)(0.8,-0.1)
\uput[0](-1,1.3){$S_3$}
\psframe[linestyle=dotted,linewidth=1pt,framearc=.5](3.2,-0.2)(1.2,0.5)
\uput[180](3.2,0.2){$S_4$}
\psframe[linestyle=dotted,linewidth=1pt,framearc=.5](3,1.3)(1.2,3)
\uput[180](3,1.7){$S_2$}
\uput[0](1.2,1.3){$P$}
\endpspicture
\end{center}
  \end{minipage}
}%
Let $P_1,\ldots,P_k$ be the paths in $\mathcal{P}[u,v]$
not covered by any bubble.  For each
$j \in \{1,\ldots,k\}$, let $\mathcal{C}_j$ be the set of tight cuts that
$P_j$ crosses an odd number of times, three or more, and let $m_j$ be
the sum over all cuts $C' \in \mathcal{C}_j$ of the number of times
that $P_j$ crosses $C'$.  By Observation~\ref{obs:odd-cross}, each
$\mathcal{C}_j$ is nonempty.  We choose a path $P = P_j$ such that
$m_j$ is minimal. Let $C$ be a cut in $\mathcal{C}_j$ (therefore $P$
crosses $C$ at least three times), and let $S_1$, $S_2$, $S_3$ and
$S_4$ be the first four connected components in order of
$P\setminus\delta_G(C)$, with $S_1,S_3\subseteq C$ (see figure).
Since $C$ is central, $S_1$ and $S_3$ are connected by a path $P_{13}$
inside of $C$, and $S_2$ and $S_4$ are connected by a path $P_{24}$
outside of $C$. Let $a$ be the endpoint of $P_{24}$ in $S_2$, let $b$
be the endpoint of $P_{13}$ in $S_3$, and let $v'$ be the endpoint of
$P_{24}$ in $S_4$. Note that there are three vertex-disjoint paths
from $a$ to $b$, and so $(a,b)$ is a $2$-vertex-cut separating $u$
from $v'$, for otherwise $G$ would have a $K_4$ minor.

The proof proceeds using a sequence of claims. The following arguments
use $C$, $u$ and $b$, but apply symmetrically to $V\setminus C$, $v'$
and $a$.
\begin{claim}\label{clm:ub0}
  Any path from $u$ to $b$ inside $C$ must cross some tight cut at
  least twice more than $P$.
\end{claim}
\begin{proof}
  If a path from $u$ to $b$ crosses no tight cut more than once, then we
  shortcut $P$ with that path, and get a simple path $P'$ from $u$ to
  $v$ that does not cross any bubble. Therefore $P' = P_\ell$ for some
  $\ell \in \{1,\ldots,k\}$. Now $P_\ell$ crosses $\delta_G(C)$ twice less than
  $P$, and does not cross any other tight cut more times than $P$;
  therefore $m_\ell < m_j$, contradicting the minimality of $m_j$.
\end{proof}

Recall
that $(a,b)$ is a $2$-vertex-cut separating $u$ from $v'$. Let $S_u$
and $S_v$ be the connected components of $V\setminus\{a,b\}$
containing $u$ and $v'$, and let $S_u^*=S_u\cup\{a,b\}$ and
$S_v^*=S_v\cup\{a,b\}$.
For subsets $S,C\subseteq V$ and vertices $u,b\in C$, we say that
\emph{$S$ separates $u$ from $b$ inside $C$} if $u$ and $b$ are in two
different connected components of $C\setminus S$.
% \begin{lemma}\label{lem:sep}
%   Suppose $S\subseteq C\setminus\{u,b\}$ separates $u$ from $b$
%   inside $C$. Then $\sigma(S)>0$; i.e. $S$ does not define a tight cut.
% \end{lemma}
% \begin{proof}
%   Suppose to the contrary that $\sigma(S)=0$. Let $B$ be the connected
%   component of $C\setminus S$ containing $b$, and let $U=C\setminus
%   (S\cup B)$. We have $\sigma(U\cup S\cup B)=\sigma(C)=0$, since $C$
%   is tight.  By Lemma~\ref{lem:easy1}$\langle c \rangle$, we have
%   $\sigma(U\cup S)+\sigma(B\cup S)=\sigma(U\cup S\cup
%   B)+\sigma(S)+2\sigma(U,B)\leq 0$, because there are no supply edges
%   from $U$ to $B$. And so, by the cut condition, we have $\sigma(B\cup
%   S)=0$, which contradicts the assumption that $P$ crosses no bubble.
% \end{proof}

\begin{claim}\label{clm:new}
  There is a $2$-vertex-cut $(x,y)$ in $S_u^*$, with not both $x$ and
  $y$ in $\{a,b\}$, with two vertex-disjoint paths $Q_1$, $Q_2$ from
  $x$ to $y$, with $Q_1\setminus\{x,y\}$ and $Q_2\setminus\{x,y\}$ not
  containing $a$ or $b$, and a demand $i$ from $Q_1\setminus\{x,y\}$
  to $Q_2\setminus\{x,y\}$.
\end{claim}
\begin{proof}
  Since $C$ is central, $u$ and $b$ are connected inside $C$, and by
  Claim~\ref{clm:ub0}, any path from $u$ to $b$ crosses some tight cut
  at least twice more than $P$. Either there is a single tight cut
  crossed by all such paths, or there is not. We prove these separate
  cases in Claim~\ref{clm:ub2} and Claim~\ref{clm:ub1} respectively.
\end{proof}

\begin{claim}\label{clm:ub2}
  If all paths from $u$ to $b$ inside $C$ cross twice the same tight
  cut, then there is a vertex $x\in Q$ separating $u$ from $b$ inside
  $C$, and a vertex $y\in P\setminus C$, such that $x$ and $y$ are
  connected by two vertex-disjoint paths $Q_1$ and $Q_2$ that do not
  contain $b$, with a demand edge going from some vertex in $Q_1
  \setminus\{x,y\}$ to some vertex in $Q_2 \setminus\{x,y\}$. Either
  $y$ is $a$, or $y$ is in the connected component of
  $V\setminus\{a,b\}$ that contains $u$.
\end{claim}
\begin{proof}%
\vspace{-1.15\baselineskip} % Undo the blank para inserted by \parpic
\parpic[r]{
  \begin{minipage}{2.6cm}
\begin{center}
\psset{unit=0.8cm,arrows=-,shortput=nab,linewidth=0.5pt,arrowsize=2pt 5,labelsep=2pt}
\pspicture(-0.7,-0.2)(2,4.4)
\dotnode(0,3.5){u}
\dotnode(2,2){a}
\dotnode(0,0.5){b}
\dotnode(2,0){v}
\dotnode(0,2){x}
\dotnode(0.9,2){p}
\nput{135}{u}{$u$}
\nput{45}{a}{$y$}
\nput{-135}{b}{$b$}
\nput{180}{x}{$x$}
\nput{-45}{p}{$p$}
\nput{-45}{v}{$v'$}
\ncline{u}{a}
\ncline{u}{x}
\ncline{b}{x}
\ncline{p}{x}
\ncline{p}{a}
\ncline{b}{a}
\ncline{b}{v}
\psframe[linestyle=dotted,linewidth=1pt,framearc=.5](-1.1,1.65)(1.5,2.35)
\psline[linestyle=dotted,linewidth=1pt](0.6,0)(0.6,4)
\uput[180](0.6,4){$C$}
\uput[0](-1.1,2){$S$}
\endpspicture
\end{center}
  \end{minipage}
}%
  \hspace{3ex}
  Let $S$ be the central set defining the tight cut crossed twice by
  all paths, with $S$ containing neither $u$ nor $b$. Since $S$ is not
  crossed by $P$ on the way from $u$ to $b$, $S$ does not contain $a$;
  and since the pair $(a,b)$ is a $2$-vertex-cut separating $u$ from
  $v'$, $S$ does not contain $v'$.
  % By Lemma~\ref{lem:sep}, if $S\subset C$, then $S$ is not tight. So
  % $S$ contains a vertex outside of $C$.

  Let $U$ be the connected component of $C\setminus S$ containing
  $u$. Let $B=C\setminus(S\cup U)$. By Lemma~\ref{lem:easy1}$\langle e \rangle$,
  $\sigma(U)-2\sigma(S,U)=\sigma(S\cup U)-\sigma(S)\geq 0$. Then
  by Lemma~\ref{lem:easy1}$\langle a \rangle$,
  \begin{equation}\label{eq:ub2-1}
    \sigma(U)\geq 2\sigma(S,U)=2\sigma(S\cap C,U)+2\sigma(S\setminus C,U).
  \end{equation}
  Since $P$ is not covered by a bubble, $\sigma(C\setminus
  U)>0$. By Lemma~\ref{lem:easy1}$\langle e \rangle$, since $U\subseteq C$, $\sigma(U)-2\sigma(C\setminus
  U,U)=\sigma(C)-\sigma(C\setminus U)< 0$. Then by Lemma~\ref{lem:easy1}$\langle a \rangle$,
  \begin{equation}\label{eq:ub2-2}
    \sigma(U)< 2\sigma(C\setminus U,U)=2\sigma(S\cap C,U)+2\sigma(B,U)
  \end{equation}
  Subtracting~\eqref{eq:ub2-1}~from~\eqref{eq:ub2-2}, we get that
  $\sigma(S\setminus C,U)<\sigma(B,U)$. Since there is no supply edge
  from $U$ to $B$, $\sigma(B,U)\leq 0$, which proves that there is a
  demand from some vertex $q\in U$ to some vertex $p\in S\setminus
  C$. Since the subpath of $P$ from $u$ to $b$ has vertices outside
  $C$, there is a path $Q'$ connecting $p$ to $P$ outside of $C$. Let
  $y$ be the endpoint of $Q'$ in $P$. Since $Q$ intersects $S$, there
  must be a path $Q''$ connecting $p$ to $Q$ inside $S$. Let $x$ be
  the endpoint of $Q''$ in $Q$. The paths $P$ and $Q$ form a cycle
  containing the vertices $u$, $y$, $b$ and $x$. The paths $Q'$ and
  $Q''$ form a path from $x$ to $y$ disjoint from that cycle, and so
  there are three vertex-disjoint paths from $x$ to $y$, and $(x,y)$
  is a $2$-vertex-cut separating $u$, $p$ and $b$, otherwise there
  would be a $K_4$. So $x$ separates $u$ from $b$ in $C$.

  Recall that there is a demand from $p$ to some vertex $q\in U$. By
  Lemma~\ref{lem:sut}, there is a simple path $Q_1$ from $x$ to $y$
  containing $p$, and a simple path $Q_2$ from $x$ to $y$ containing
  $q$. The paths $Q_1$ and $Q_2$ must be vertex-disjoint; otherwise
  there would be a path from $p$ to $q$ disjoint from $\{x,y\}$, and
  since $U$ is a connected component of $C\setminus S$, a path inside
  $U$ from $q$ to $u$ disjoint from $\{x,y\}$, contradicting the fact
  that $(x,y)$ separate $p$ from $u$. Finally, $(a,b)$ is a
  $2$-vertex-cut, so since there is a path from $u$ to $y$ through $x$
  disjoint from $\{a,b\}$, $y$ cannot be in a different connected
  component of $V\setminus\{a,b\}$ than $u$.
\end{proof}

\begin{claim}\label{clm:ub1}
  If there is no single tight cut crossed twice by all paths
  connecting $u$ and $b$ inside $C$, then there is a demand edge going
  from one of those paths to another.
\end{claim}
\begin{proof}%
\vspace{-1.15\baselineskip} % Undo the blank para inserted by \parpic
\parpic[r]{
  \begin{minipage}{2.6cm}
\begin{center}
\psset{unit=0.8cm,arrows=-,shortput=nab,linewidth=0.5pt,arrowsize=2pt 5,labelsep=2pt}
\pspicture(-0.7,0)(3.3,3.7)
\dotnode(0,3){u}
\dotnode(2,2.5){a}
\dotnode(0,0.5){b}
\dotnode(2,0){v}
\dotnode(-0.5,1.75){x1}
\dotnode(0,1.75){x2}
\dotnode(0.5,1.75){x3}
\nput{135}{u}{$u$}
\nput{45}{a}{$a$}
\nput{-135}{b}{$b$}
\nput{-45}{v}{$v'$}
\ncline{u}{a}
\ncline{u}{x1}
\ncline{u}{x2}
\ncline{u}{x3}
\ncline{b}{x1}
\ncline{b}{x2}
\ncline{b}{x3}
\ncline{b}{a}
\ncline{b}{v}
\psframe[linestyle=dotted,linewidth=1pt,framearc=.5](-1.2,1.4)(0.7,2.1)
\psline[linestyle=dotted,linewidth=1pt](1,0)(1,3.5)
\uput[180](1,3.5){$C$}
\uput[0](-1.2,1.75){$S$}
\endpspicture
\end{center}
  \end{minipage}
}%
  \hspace{3ex}
For every path connecting $u$ to $b$ inside $C$, we choose a tight cut
crossed twice, and we contract all edges of the path that do not cross
that tight cut. Each of the paths now has two edges. Let $S$ denote
the set of vertices in the middle of these paths. There are no supply
edges from a vertex in $S$ to any vertex except $u$ or $b$ because
that would create a $K_4$ minor. Since every path from $u$ to $b$
crosses some central tight cut twice, every vertex in $S$ defines a
bubble separating $u$ from $b$ inside $C$. The supply graph induced by
$u$, $b$ and vertices in $S$ is a $K_{2,m}$.
By assumption, there is no single tight cut crossed twice by all
paths, so $\sigma(S)>0$, even though every vertex in $S$ defines a
tight cut. And so, by Lemma~\ref{lem:easy1}$\langle e \rangle$, there
must exist demands between vertices of $S$.
%   For every path connecting $u$ to $b$ inside $C$, we choose a tight
%   cut crossed twice, and we contract all edges of the path that do not
%   cross that tight cut. Each of the paths now has two edges. Let $S$
%   denote the set of vertices in the middle of these paths. There are
%   no supply edges from a vertex in $S$ to any vertex except $u$ or $b$
%   because that would create a $K_4$ minor. Since every path from $u$
%   to $b$ crosses some central tight cut twice, every vertex in $S$
%   defines a bubble separating $u$ from $b$ inside $C$. The supply
%   graph induced by $u$, $b$ and vertices in $S$ is a $K_{2,m}$.

%   Let $B$ be the connected component of $C\setminus S$ containing $b$,
%   and let $U=C\setminus (S\cup B)$. By Lemma~\ref{lem:sep}, $S$ does
%   not define a tight cut, even though every vertex in $S$ defines a
%   tight cut. And so, by Lemma~\ref{lem:easy1}$\langle e \rangle$,
%   there must exist demands between vertices of $S$.
\end{proof}

\begin{figure}[ht]
\begin{center}
    \psset{unit=0.8cm,xunit=1cm,arrows=-,shortput=nab,linewidth=0.5pt,arrowsize=2pt 5,labelsep=1.5pt}
    \pspicture(0,0.8)(4,3.5)
    \dotnode(1,3.3){x}
    \dotnode(3,3.3){y}
    \dotnode(2,4){d1}
    \dotnode(2,2.6){d2}
    \dotnode(1,0.7){xp}
    \dotnode(3,0.7){yp}
    \dotnode(2,0){d1p}
    \dotnode(2,1.4){d2p}
    \dotnode(4,2){b}
    \dotnode(0,2){a}
    \nput{135}{x}{$x$}
    \nput{45}{y}{$y$}
    \nput{-135}{xp}{$x'$}
    \nput{-45}{yp}{$y'$}
    \nput{-135}{a}{$a$}
    \nput{-45}{b}{$b$}
    \ncarc{a}{x}
    \ncarc{xp}{a}
    \ncarc{y}{b}
    \ncarc{b}{yp}
    \ncarc{x}{d1}
    \ncarc{d1}{y}
    \ncarc{d2}{x}
    \ncarc{y}{d2}
    \ncarc{d1p}{xp}
    \ncarc{yp}{d1p}
    \ncarc{xp}{d2p}
    \ncarc{d2p}{yp}
    \ncline{a}{b}
    \ncline[linestyle=dashed]{d1}{d2}
    \ncline[linestyle=dashed]{d1p}{d2p}
    \uput[0](2,3.3){$i$}
    \uput[45](2.4,3.8){$Q_1$}
    \uput[-45](2.5,2.9){$Q_2$}
    \uput[0](2,0.7){$i'$}
    \psline[linestyle=dotted,linewidth=1pt](-0.2,2.2)(4.4,2.2)
    \uput[90](4.2,2.2){$S_u$}
\endpspicture
\end{center}
\caption{Subgraph showing the relations of the $2$-vertex-cut $(a,b)$,
  $(x,y)$ and $(x',y')$. One of $x$ or $y$ may be $a$ or $b$, but not
  both. One of $x'$ or $y'$ may be $a$ or $b$, but not both. The
  demands $i$ and $i'$ are dashed.}\label{fig:ab}
\end{figure}

Note that Claim~\ref{clm:new} also applies to $S_v^*$, and so there is
in $S_v^*$ a $2$-cut $(x',y')$, with two vertex-disjoint paths from
$x'$ to $y'$, and a demand $i'$ connecting these two paths (see
Figure~\ref{fig:ab}).

Recall that $(s,t)$ is a split pair. Since $(a,b)$ is a $2$-vertex-cut
connected by three disjoint paths, $s$ and $t$ cannot be in different
connected components of $V\setminus\{a,b\}$, because otherwise an
$(s,t)$ edge would create a $K_4$. So at least one of $S_u$ or $S_v$
contains neither $s$ nor $t$.

\begin{claim}
  Suppose $S_u$ contains neither $s$ nor $t$. Then $i$ is
  non-compliant; the pair $(x,y)$ which separates its endpoints is its
  pair of terminals, and this pair of terminals is strictly bracketed
  by the pair $(w,z)$ of terminals of $(u,v)$.
\end{claim}
\begin{proof}
  By Lemma~B.4, for any $v'\in S_u$, there is a simple path from $s$ to
  $t$ containing $v'$, so there is a simple path from $s$ to $t$ that
  goes through $S_u$, and so contains $a$ and $b$. Without loss of
  generality, assume that the path meets $a$ before $b$ on the way from
  $s$ to $t$. Then since the orientation is acyclic, there is no simple
  path from $s$ to $t$ that meets $b$ before $a$.  Since any edge in $G$
  is oriented in the direction it appears on any simple path from $s$ to
  $t$, then any edge in $S_u^*$ is oriented in the direction it appears
  on any simple path from $a$ to $b$. So $a$ is the unique source in
  $S_u^*$, and $b$ the unique sink. Any simple path from $a$ to $b$
  through an endpoint of $i$ contains $x$ and $y$, and does not contain
  the other endpoint of $i$. So $i$ is a non-compliant demand, and
  $(x,y)$ is its pair of terminals, which is bracketed by $(a,b)$. Note
  that $(x,y)$ is not the same as $(a,b)$.

\begin{figure}[ht]
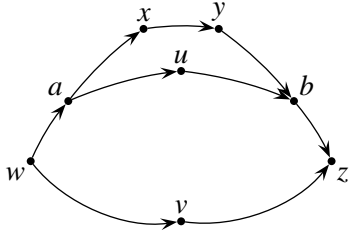

\begin{center}
    \psset{unit=0.8cm,xunit=1cm,arrows=->,shortput=nab,linewidth=0.5pt,arrowsize=2pt 5,labelsep=1.5pt}
    \pspicture(0,0.8)(4,3.7)
    \dotnode(0,1){w}
    \dotnode(4,1){z}
    \dotnode(0.5,2){a}
    \dotnode(3.5,2){b}
    \dotnode(2,2.5){u}
    \dotnode(2,0){v}
    \dotnode(1.5,3.2){x}
    \dotnode(2.5,3.2){y}
    \nput{-135}{w}{$w$}
    \nput{-45}{z}{$z$}
    \nput{135}{a}{$a$}
    \nput{45}{b}{$b$}
    \nput{90}{u}{$u$}
    \nput{90}{v}{$v$}
    \nput{90}{x}{$x$}
    \nput{90}{y}{$y$}
    \ncarc{w}{a}
    \ncarc{a}{u}
    \ncarc{u}{b}
    \ncarc{b}{z}
    \ncarc{a}{x}
    \ncarc{x}{y}
    \ncarc{y}{b}
    \ncarc[arcangle=-30]{w}{v}
    \ncarc[arcangle=-30]{v}{z}
\endpspicture
\end{center}
\caption{Relative positions of $w$, $z$, $a$, $b$, $x$, $y$, $u$ and
  $v$. It is possible that $w=a$, or $z=b$. The vertex $u$ may be on
  the path from $a$ to $b$ containing $x$ and $y$, but does not need
  to be.}\label{fig:wz}
\end{figure}

We prove that $(a,b)$ is bracketed by the pair $(w,z)$ of terminals of
$(u,v)$, which means that $(x,y)$ is bracketed by $(w,z)$. By
Lemma~\ref{lem:terminals}, any cycle $C$ containing $u$ and $v$ also
contains the terminals $w$ and $z$ of the demand $(u,v)$, and is
composed of two oriented paths from one terminal to the other, say
from $w$ to $z$, and $w$ is the unique source of $C$ and $z$ its
unique sink. The cycle $C$ must contain $a$ and $b$ since $(a,b)$ is a
$2$-vertex-cut separating $u$ from $v$. Since any simple path from $a$
to $b$ in $S_u^*$ is oriented from $a$ to $b$, the part of $C$ in
$S_u^*$ is oriented from $a$ to $b$. So neither $w$ nor $z$ is in
$S_u$, because then they would not be source or sink of $C$. So $C$
contains a path $Q$ from $w$ to $z$ through $u$, and $Q$ contains $a$
and $b$; so $(a,b)$ is bracketed by $(w,z)$. So $(x,y)$ is bracketed
by $(w,z)$.
\end{proof}
Since at least one of $S_u$ and $S_v$ contains neither $s$ nor $t$, at
least one of $(x,y)$ and $(x',y')$ is bracketed by $(w,z)$,
contradicting our choice of $(u,v)$. \hfill\hfill\hfill\hfill\hfill\hfill\emph{This completes the proof of
  Lemma~\ref{lem:existbubbles}.}
\end{proof}

% ---> Integral routability theorem; the associated algorithm <---

\section{Integrally Routable Series-Parallel Instances}\label{sec:integral}
In this section, we prove Theorem~\ref{thm:integral}, which we restate
here:

\begin{theorem}\label{thm:integral2}
  Let $(G,H,c,D)$ form an instance of the multicommodity flow problem,
  such that $G$ is series-parallel, $(G,H)$ is cut-sufficient, and
  $(G,H,c,D)$ is Eulerian. Then the instance has an integral solution
  if and only if it satisfies the cut condition, and that integral
  solution can be computed in polynomial-time.
\end{theorem}
Since an instance that does not satisfy the cut condition cannot have
a solution, integral or otherwise, we only need to prove the other
direction.

For any demand $d=(u,v)$ and vertex $w$ in a multiflow instance,
\emph{pushing a unit of $d$ to $w$} consists of removing one unit of
demand $d$, and creating two demand edges of unit demand from $u$ to
$w$ and $w$ to $v$. This can be seen as taking the decision of routing
at least one unit of the demand $d$ through $w$.

For any demand $d$ whose endpoints are connected by a path $P$,
\emph{routing a unit of $d$ along $P$} consists of removing one unit of
capacity along each edge of $P$, and removing one unit of demand from
$d$. Supply edges whose capacity falls to zero are removed from $G$, and
demand edges whose demand falls to zero are removed from $H$. For each
$S\subseteq V$, define $n_S=|\delta_G(S)\cap P|$. The operation reduces
the surplus $\sigma(S)$ by $2\lfloor n_S/2\rfloor$: it reduces the total
of capacities crossing $\delta_G(S)$ by $n_S$; and if $n_S$ is odd, then
$d\in\delta_H(S)$ and it reduces the total demand crossing $\delta_H(S)$
by $1$. Thus, the surplus of any cut is reduced by an even number.

Suppose we are given a series-parallel instance that is cut-sufficient, 
Eulerian, and satisfies the cut condition, with a demand $d=(u,v)$.  We prove 
that (A) there is a sequence of push operations to move a unit of demand $d$
to a path $Q$ of unit demands from $u$ to $v$ without breaking the cut
condition; and (B) the unit demands in $Q$ can all be
routed without breaking the cut condition. Thus, the demands in $Q$ fall to
zero, and are removed. The two operations are equivalent to
routing one unit of $d$; thus, we get a smaller instance which has the
same properties. We can therefore recursively build a solution to the
whole problem.

We embed $G$ in the plane such that the endpoints $u$ and $v$ are on
the outside face.  (Lemma~\ref{lem:sepa-embed})
Any path $P$ from $u$ to $v$ thus partitions $G\setminus P$ into two sides, 
one to the left and on to the right of $P$. Two paths 
$P$ and $P'$ \emph{cross} if $P'$ contains vertices on both sides of $P$. 
We decompose the flow of the fractional solution routing the demand 
$d$ into paths in the series-parallel supply graph such that no two 
paths cross. This gives an ordering of the path $P_1,\ldots,P_k$ such 
that if $P_1$ and $P_k$ have a common vertex, then all paths $P_j$, 
$j=1,\ldots,k$ go through that vertex. We examine the subgraph 
$P_1\cup P_k$. Since $u$ and $v$ are on the outside face of $G$, the 
graph $P_1\cup P_k$ is composed of a family of cycles (whenever $P_1$ 
and $P_k$ are disjoint) connected by paths (whenever $P_1$ and $P_k$
coincide). Let $C_1,\ldots,C_j$ be the cycles in $P_1\cup P_k$, and
for any cycle $C_i$, let $a_i$ and $b_i$ be the two vertices of $C_i$
contained in both $P_1$ and $P_k$.  See Figure~\ref{fig:p1pk}.

\begin{figure}
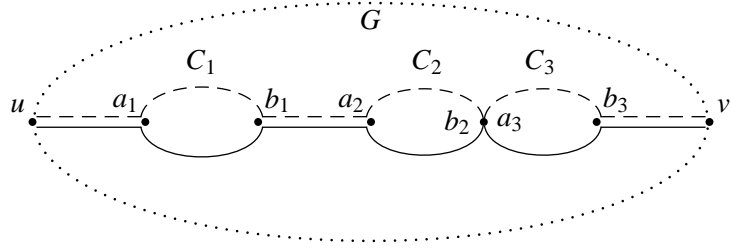

\begin{center}
    \psset{unit=0.8cm,xunit=1cm,arrows=-,shortput=nab,linewidth=0.5pt,arrowsize=2pt 5,labelsep=3.5pt}
    \pspicture(0,0.8)(9,3.5)
    \psellipse[linestyle=dotted,linewidth=1pt](4.5,2)(4.5,2)
    \rput(4.5,3.7){$G$}
    \rput(2.25,3){$C_1$}
    \rput(5.25,3){$C_2$}
    \rput(6.75,3){$C_3$}
    \dotnode(0,2){u}
    \dotnode(1.5,2){a}
    \dotnode(3,2){b}
    \dotnode(4.5,2){c}
    \dotnode(6,2){d}
    \dotnode(7.5,2){e}
    \dotnode(9,2){v}
    \nput{130}{u}{$u$}
    \nput{50}{v}{$v$}
    \nput{130}{a}{$a_1$}
    \nput{50}{b}{$b_1$}
    \nput{130}{c}{$a_2$}
    \nput{180}{d}{$b_2$}
    \nput{0}{d}{$a_3$}
    \nput{50}{e}{$b_3$}
    \psset{offset=2pt,linestyle=dashed}
    \ncline{u}{a}
    \ncarc[arcangle=80]{a}{b}
    \ncline{b}{c}
    \ncarc[arcangle=80,offsetB=0pt]{c}{d}
    \ncarc[arcangle=80,offsetA=0pt]{d}{e}
    \ncline{e}{v}
    \psset{offset=-2pt,linestyle=solid}
    \ncline{u}{a}
    \ncarc[arcangle=-80]{a}{b}
    \ncline{b}{c}
    \ncarc[arcangle=-80,offsetB=0pt]{c}{d}
    \ncarc[arcangle=-80,offsetA=0pt]{d}{e}
    \ncline{e}{v}
\endpspicture
\end{center}
\caption{Illustration of the planar embedding with $u$ and $v$ on the
  outside face. The dotted cycle represents the outside face of $G$,
  the paths $P_1$ is in solid, and the path $P_k$ in
  dashed.}\label{fig:p1pk}
\end{figure}

\begin{lemma}\label{lem:push}
  In any instance of the multiflow problem satisfying the cut
  condition, if there is a fractional solution such that all paths
  $P_1,\ldots,P_k$ routing demand $d=(u,v)$ go through the same vertex
  $w$, then it is possible to push a unit of the demand $d$ to the
  vertex $w$ without breaking the cut condition.
\end{lemma}
\begin{proof}
  Let $\mathcal{C}_{uv,w}$ be the set of cuts separating $u$ and $v$
  from $w$. If we push a unit of $d$ to $w$, only the surpluses of
  cuts in $\mathcal{C}_{uv,w}$ are modified, and each surplus is
  reduced by two units. It is thus sufficient to prove that all cuts
  in $\mathcal{C}_{uv,w}$ have a surplus of at least two.

  We execute the following operations on the multiflow problem and its
  fractional solution. We reduce the demand of $d$ by one unit. Let
  $f_1,\ldots,f_k$ be the flows of the fractional solution routed on
  paths $P_1,\ldots,P_k$. We chose quantities $0\leq g_i\leq f_i$,
  $i=1,\ldots,k$, such that $\sum_ig_i=1$. We remove successively
  from each edge in $P_i$ a quantity $g_i$ of capacity, and subtract
  $g_i$ from $f_i$, with $i=1,\ldots,k$.
%   \footnote{All these operations
%     may be summarized as routing one unit of $d$, using part of the
%     fractional solution.}

  Since each path $P_i$ crosses every cut in $\mathcal{C}_{uv,w}$ at
  least twice, these operations reduce the surplus of every cut in
  $\mathcal{C}_{uv,w}$ by at least two. The remainder flow of
  $f_1,\ldots,f_k$ on paths $P_1,\ldots,P_k$ gives a fractional
  solution routing the reduced demand, and so the instance still
  satisfies the cut condition. So for each $S\in\mathcal{C}_{uv,w}$,
  $\sigma(S)\geq 0$ after $\sigma(S)$ has been reduced by at least
  two, so $\sigma(S)$ was at least two in the original instance.
\end{proof}

We push the demand $d$ to every vertex in $P_1\cap P_k$. By
Lemma~\ref{lem:push}, we can do this without breaking the cut
condition, since all paths routing $d$ in the fractional solution go
through these vertices. This creates a path $Q$ of unit demands from
$u$ to $v$, such that the vertices of $Q$ are the vertices in both
$P_1$ and $P_k$.  This completes part (A).  

We next argue that we can route the demands in $Q$.  We need to identify
paths in $G$ to do this.
The path $Q$ has a unit demand parallel to every edge
in the paths connecting the cycles $C_1,\ldots,C_j$, and a unit demand
from $a_i$ to $b_i$ for every cycle $C_i$, $i=1,\ldots,j$.  We will
route the demands in $Q$ along the paths connecting the cycles, and
then along one side of each cycle.  The side we pick is guided
by the next two lemmas.

For any cycle $C_i$ not containing $v$, we say a vertex $w\in C_i$ is
\emph{linked to $v$} if there is in $G$ a path from $w$ to $v$
containing only the vertex $w$ in $C_i$.
\begin{lemma}
\label{lem:link}
  Any cycle $C_i$ not containing $v$ contains at most one
  vertex apart from $b_i$ that is linked to $v$.
\end{lemma}
\begin{proof}
  Contract the connected component of $G\setminus C_i$ containing $v$.
  The resulting vertex is connected by an edge to any vertex of $C_i$
  that is linked to $v$. If there are three, this forms a $K_4$.
\end{proof}

We define the path $P$ from $u$ to $v$ by choosing for each cycle
$C_i$ the side of $C_i$ from $a_i$ to $b_i$ that does not contain a
vertex linked to $v$. This is always possible by Lemma~\ref{lem:link}.

If $\delta_G(S)$ is a central cut, by Lemma~\ref{lem:cycle} it crosses
a cycle $C_i$ either twice, or not at all.  For any $u$-to-$v$ path
$P'$ in $P_1\cup P_k$ obtained by choosing for each cycle $C_i$
either $C_i \cap P_1$ or $C_i \cap P_k$, the
cut $\delta_G(S)$ crosses $P'\cap C_i$ zero, once, or twice for every
$i=1,\ldots,j$.  Our choice of $P$ given above is special:

\begin{lemma}\label{lem:twice}
  For any central cut $\delta_G(S)$, there is at most one cycle $C_i$ such that
  $\delta_G(S)$ crosses $P\cap C_i$ twice.
\end{lemma}
\begin{proof}
  Suppose that there is a set $S$ defining a cut that crosses $P\cap
  C_i$ twice and $P\cap C_l$ twice, for $i<l$. Then $S$ either
  contains both $a_i$ and $b_i$, or neither of them. Suppose without
  loss of generality that it contains neither. Then $S$ contains some
  vertices in $P\cap C_i\setminus\{a_i,b_i\}$. Since $\delta_G(S)$ also
  intersects $C_l$, the set $S$ also contains some vertex in $C_l$.
  As $S$ is central, there must be a path from $(P\cap
  C_i)\setminus\{a_i,b_i\}$ to $C_l$, which means that some vertex of
  $(P\cap C_i)\setminus\{a_i,b_i\}$ is linked to $v$. This contradicts
  our choice of $P$.
\end{proof}

\begin{lemma}\label{lem:unit}
  We can route the unit demands in $Q$ along the path $P$ without
  breaking the cut condition.
\end{lemma}
\begin{proof}
  The path $P$ goes through both extremities of every demand we
  created by pushing $d$. Routing any demand parallel to a supply edge
  consists of removing one unit of capacity from the supply edge and
  removing the unit demand. The surplus of any cut crossing such a
  demand is not affected by this. Routing a demand across a cycle
  $C_i$, from $a_i$ to $b_i$, consists of removing one unit of
  capacity of each supply edge in $P\cap C_i$, and removing the unit
  demand. If a central cut $\delta_G(S)$ separates $a_i$ from $b_i$,
  it crosses $P\cap C_i$ exactly once, and so its surplus $\sigma(S)$
  is not affected by this. If a central cut $\delta_G(S)$ does not
  separate $a_i$ from $b_i$, then its surplus is reduced by two or
  unchanged, depending on whether it crosses $P\cap C_i$ twice or not at
  all. For any central cut $\delta_G(S)$, there is at most one cycle
  $C_i$ such that $\delta_G(S)$ crosses $P\cap C_i$ twice, by
  Lemma~\ref{lem:twice}. So the surplus of any cut is reduced at most
  by two. As there is a positive flow routing demand $d$ along path
  $P$ in the fractional solution, no cut that crosses $P$ more than
  once is tight: because in any solution to the multiflow problem, the
  supply edges crossing a tight cut have their capacity completely
  used to route the demands that also cross it. As the instance is
  Eulerian, any cut that is not tight has a surplus of at least
  two. And so routing one unit along $P$ does not break the cut
  condition.
\end{proof}

The flow routing all the demands created by pushing $d$ is also a way
of routing one unit of $d=(u,v)$ in the original problem; so we have
found a path $P$ from $u$ to $v$ such that routing one unit of $d$
along this path does not break the cut condition. After doing this,
the reduced instance still does not have any odd spindle as a minor,
since no demand edges were introduced; is still Eulerian, and still
satisfies the cut condition. By induction, we can find an integral
routing for the instance.

\subsection{Polynomial-Time Algorithm}
\label{sec:alg}
The method described in the proof of Theorem~\ref{thm:integral2}
routes one unit of flow at a time. We first show that each unit can be
routed in polynomial-time.  This gives us a pseudo-polynomial-time
algorithm for an instance $(G,H,c,D)$; the algorithm is polynomial in
the size of $G=(V,E)$, $H=(V,F)$ and the bit-size of $c$, but only
polynomial in $D$, the demands assigned to edges of $H$, instead of in
the \emph{bit-size} of $D$. We then give a fully-polynomial-time
algorithm, that reduces the instance to another one in which $D$ is
polynomial in the size of $G$ and $H$, and then uses the
pseudo-polynomial-time algorithm.

First, it is possible to find a fractional solution to the problem in
polynomial-time by linear programming. The problem can indeed be
solved by a polynomial-sized linear program, by having one variable
$f^i_e$ indicating the amount of commodity $i$ flowing through edge
$e$, for every $i\in F$ and $e\in E$ (e.g. Section 70.6 of
\cite{SchrijverBook}). This linear program can be then solved
efficiently in polynomial time using interior point methods.

The second step is to embed the planar graph $G'$ into the plane.
This can be done in time linear in the number of
vertices~\cite{Schnyder}.

We then decompose the flow of the fractional solution routing a demand
into paths $P_1,\ldots,P_k$. Let $m=|E|$. The flow decomposition has
$k\leq m$ paths, and can be found in $O(m^2)$ time, given the
fractional flow.
% We can find in $O(m)$ time each of the  $P_1,\ldots,P_k$
% (each path $P_i$ saturating the flow of the fractional solution in
% at least one edge), so this can be done in $O(m^3)$ time (and
% probably faster, but without changing the total complexity of the
% algorithm).

Finally, we find for each of the $O(m)$ cycles in $P_1\cup P_k$ which
side has a vertex linked to $v$. This can be done by an exploration
algorithm in $O(m)$ time, which makes $O(m^2)$ time in total. The
operation of routing a unit through the path $P$ is done in $O(m)$
time.

So routing one unit of demand can be done in polynomial-time, with a
theoretical complexity dominated by the resolution of the linear
program finding a fractional solution.

We now present a polynomial-time algorithm. We start by finding a
fractional solution to the problem, solving the polynomial-sized
linear program. For each demand $i\in F$, we do a path decomposition
of the flow routing $i$. This yields $k\leq m$ paths $P_1,\ldots,P_k$
per demand $i$. For each path $P$ routing a quantity $f^i_{P}$ of flow
between endpoints of $i$, we send $\lfloor f^i_{P}\rfloor$ units of
flow on $P$. After this, each path $P_j$ routes an amount of flow
smaller than $1$, and since there are no more than $m$ paths routing
each demand, we are left with at most $m|F|$ units of demand to
route. We use then the pseudo-polynomial algorithm presented
above. The theoretical complexity of the algorithm is dominated by
that of this last step, which solves at most $m|F|$ linear programs
finding a fractional solution.

% ---> Discussion; conclusions; open problems <---

\section{Discussion}\label{sec:conclusion}
In this paper, we give a complete characterization for cut-sufficient
multiflow problems in series-parallel instance. A pair $(G,H)$ is
\emph{minimally cut-insufficient} if it is not cut-sufficient, but
deleting any edge or demand or contracting any edge makes it
cut-sufficient. Since any pair that is not cut-sufficient contains a
pair that is minimally cut-insufficient as a minor, then our results
show that odd spindles are the only minimally cut-insufficient
pairs with $G$ series-parallel.

A natural extension of this result is to planar pairs, i.e., pairs
where the supply graph is planar. There are planar pairs that are not
cut-sufficient, yet do not have an odd spindle as a minor. A
\emph{bad-$K_4$-pair} is the example in Figure~\ref{fig:k4},
attributed by~\cite{SchrijverBook} to Papernov, which is of particular
interest. Apart from odd spindles, it is the only minimally
cut-insufficient pair we know of.

\begin{figure}[ht]
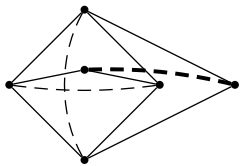

\begin{center}
\psset{unit=1cm,yunit=1cm,arrows=-,shortput=nab,linewidth=0.5pt,arrowsize=2pt 5,labelsep=1.5pt}
\pspicture(0,.3)(3,2)
\dotnode(0,1){u1}
\dotnode(1,1.2){u2}
\dotnode(2,1){u3}
\dotnode(1,0){v1}
\dotnode(3,1){v2}
\dotnode(1,2){v3}
\ncline{u1}{v1}
\ncline{u1}{v3}
\ncline{u3}{v1}
\ncline{u3}{v3}
\ncline{u1}{u2}
\ncline{u2}{u3}
\ncline{v1}{v2}
\ncline{v2}{v3}
\psset{linestyle=dashed}
\ncarc{u3}{u1}
\ncarc[arcangle=30]{v1}{v3}
\ncarc[linewidth=1.5pt]{u2}{v2}
\endpspicture
\end{center}
\caption{Planar pair without odd spindle as a minor, and not
  cut-sufficient. Supply edges are solid, and demands are dashed.  If
  the thick dashed edge has demand $2$, and all other capacities and
  demands are $1$, the instance is Eulerian and satisfies the cut
  condition, but is not routable.}\label{fig:k4}
\end{figure}

\begin{conjecture}
  Odd spindles and the bad-$K_4$-pair are the only minimally
  cut-insufficient pairs $(G,H)$, with $G$ planar.
\end{conjecture}
This would imply that a planar pair is cut-sufficient if and only if
it does not contain an odd spindle or the bad-$K_4$-pair in
Figure~\ref{fig:k4} as a minor.

% This would lead to the following result: Let $\mathcal{P}$ be a
% property on pairs, such that $\mathcal{P}$ is preserved by
% minor-taking, and such that odd-$K_{2,p}$-pairs and the bad-$K_4$-pair
% do not have $\mathcal{P}$. Then any planar pair with the property
% $\mathcal{P}$ is cut-sufficient.

\paragraph{Acknowledgments:} The third author gratefully
thanks Bruce Shepherd for many discussions.

\bibliographystyle{abbrv} \bibliography{stoc12-article}

\end{document}